%% file: main.tex
\documentclass[journal]{IEEEtran}
\usepackage{cite}
\usepackage{amsmath,amssymb,amsfonts}
\usepackage{amsthm}
\usepackage{algorithmic,algorithm}
\usepackage{graphicx}
\usepackage{textcomp}
\usepackage{url}
\usepackage{hyperref}
\usepackage{subcaption}
\usepackage{comment}

\def\BibTeX{{\rm B\kern-.05em{\sc i\kern-.025em b}\kern-.08em
    T\kern-.1667em\lower.7ex\hbox{E}\kern-.125emX}}

\newcommand{\Omegat}{\Omega_t}

\def\*#1{\mathbf{#1}}

\newcommand{\RR}{\mathbb{R}^{2}}
\newcommand{\Qi}{Q_{i}}
\newcommand{\Qj}{Q_{j}}

\newcommand{\pos}{p}
\newcommand{\posi}{p_{i}}
\newcommand{\pd}{p_d}
\newcommand{\posrel}{\tilde{p}}
\newcommand{\vel}{v}
\newcommand{\vd}{v_d}
\newcommand{\velrel}{\tilde{v}}

\newcommand{\hrel}{\tilde{h}}
\newcommand{\pix}{p_{i,x}}
\newcommand{\piy}{p_{i,y}}
\newcommand{\veli}{v_{i}}

\newcommand{\vix}{v_{i,x}}
\newcommand{\viy}{v_{i,y}}
\newcommand{\ui}{u_i}
\newcommand{\uix}{u_{i,x}}
\newcommand{\uiy}{u_{i,y}}

\newcommand{\pjx}{p_{j,x}}
\newcommand{\pjy}{p_{j,y}}
\newcommand{\vjx}{v_{j,x}}

\newcommand{\ujx}{u_{j,x}}
\newcommand{\ujy}{u_{j,y}}

\newcommand{\heading}{\theta}
\newcommand{\speed}{s}
\newcommand{\csf}{f_{al}}

\newcommand{\vxr}{v_{r,x}}
\newcommand{\vyr}{v_{r,y}}
\newcommand{\pxr}{p_{r,x}}
\newcommand{\pyr}{p_{r,y}}
\newcommand{\vdxr}{\dot{v}_{r,x}}
\newcommand{\vdyr}{\dot{v}_{r,y}}
\newcommand{\pdxr}{\dot{p}_{r,x}}
\newcommand{\pdyr}{\dot{p}_{r,y}}

\newcommand{\ttr}{\psi}
\newcommand{\rd}{r_d}

\newcommand{\tc}{t_*}

\newcommand{\aI}{a_{I}}
\newcommand{\ah}{a_{h}}
\newcommand{\av}{a_{v}}

\newcommand{\domDanger}{\Gamma_{D}}
\newcommand{\domSafe}{\Gamma_{S}}

\let\oldnl\nl
\newcommand{\nonl}{\renewcommand{\nl}{\let\nl\oldnl}}

\newtheorem{thm}{Theorem}[section]
\newtheorem{prop}[thm]{Proposition}
\newtheorem{lem}[thm]{Lemma}

\newtheorem{rmk}[thm]{Remark}
\newtheorem{defn}[thm]{Definition}

\begin{document}
\title{Safe Coverage of Moving Domains for Vehicles with Second Order Dynamics}
\author{Juan Chacon, Mo Chen, and Razvan C. Fetecau
\thanks{
J. Chacon is with the Department of Mathematics, Simon Fraser University, Burnaby, BC V5A 1S6, CANADA  (e-mail: jchaconl@sfu.ca). }
\thanks{M. Chen is with the School of Computing Science, Simon Fraser University, Burnaby,
BC V5A 1S6, CANADA (e-mail: mochen@cs.sfu.ca).}
\thanks{R.C. Fetecau is with the Department of Mathematics, Simon Fraser University, Burnaby, BC V5A 1S6, CANADA (e-mail: van@math.sfu.ca).}}

\maketitle

\begin{abstract}
Autonomous coverage of a specified area by robots operating in close proximity with each other has many potential applications such as real-time monitoring of rapidly changing environments, and search and rescue; however, coordination and safety are two fundamental challenges.
For coordination, we propose a distributed controller for covering moving, compact domains for two types of vehicles with second order dynamics (double integrator and fixed-wing aircraft) with bounded input forces. 
This control policy is based on artificial potentials and alignment forces designed to promote desired vehicle-domain and inter-vehicle separations and relative velocities.
We prove that certain coverage configurations are locally asymptotically stable. 
For safety, we establish energy conditions for collision free motion and utilize Hamilton-Jacobi (HJ) reachability theory for last-resort pairwise collision avoidance. 
We derive an analytical solution to the associated HJ partial differential equation corresponding to the collision avoidance problem between two double integrator vehicles. 
We demonstrate our approach in several numerical simulations involving the two types of vehicles covering convex and non-convex moving domains.
\end{abstract}


\input{introduction}
\input{static_domain}
\input{moving_domain}
\input{fixed_wing}
\input{conclusions}

\bibliographystyle{ieeetr}  
\bibliography{references}
\input{biography}
\end{document}

%% file: introduction.tex
\section{Introduction}

Autonomous systems have many potential applications in almost every part of society; however, these systems still typically operate in controlled environments in the absence of other agents.
Two major challenges -- coordination and safety -- arise when autonomous systems cooperate in close proximity with each other.
In this paper, we consider specifically the problem of controlling multiple autonomous systems to cover a desired possibly moving area in a decentralized and safe manner.
Applications of this problem include real-time surveillance of dynamic environments, efficient search and rescue, and multi-agent aerobatics.
\begin{figure}[ht]
\centering
\begin{subfigure}[b]{0.52\columnwidth}
\includegraphics[width=\columnwidth]{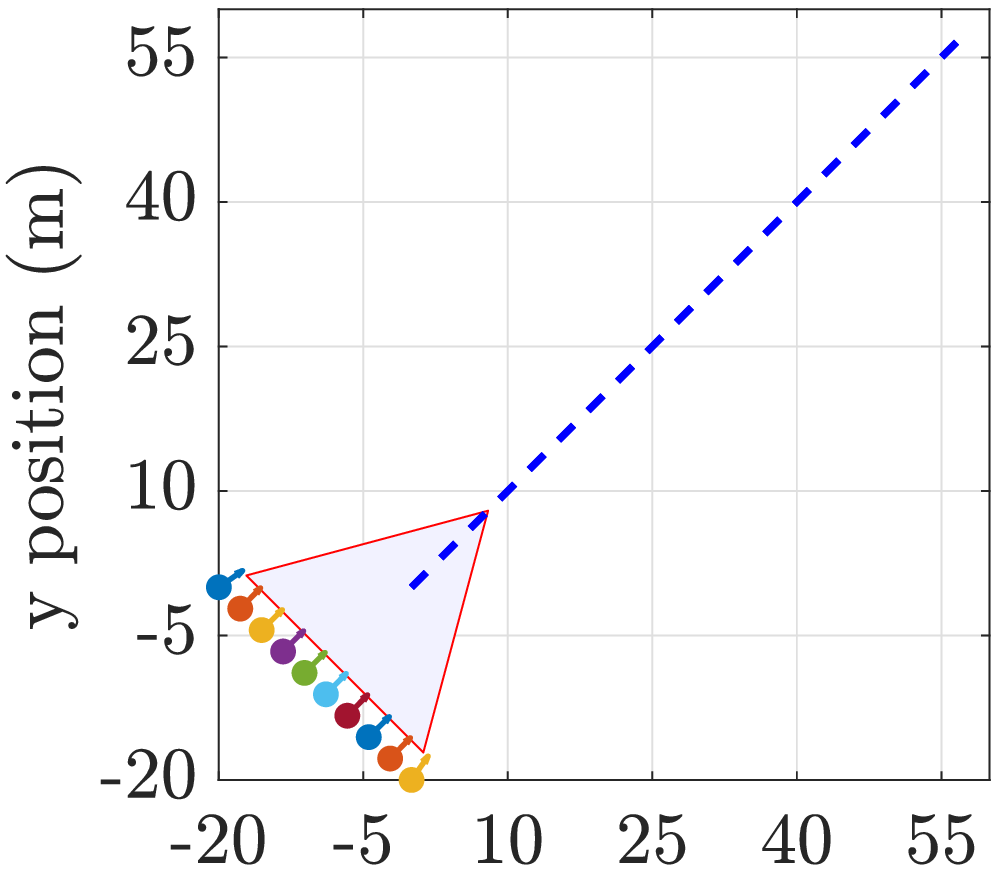}
\vspace{-2.5em}\caption{$t=0(s)$}
\label{fig:moving_triangle_avoid_0}
\end{subfigure}
\hspace{-1.9em}
\begin{subfigure}[b]{0.52\columnwidth}
\includegraphics[width=\columnwidth]{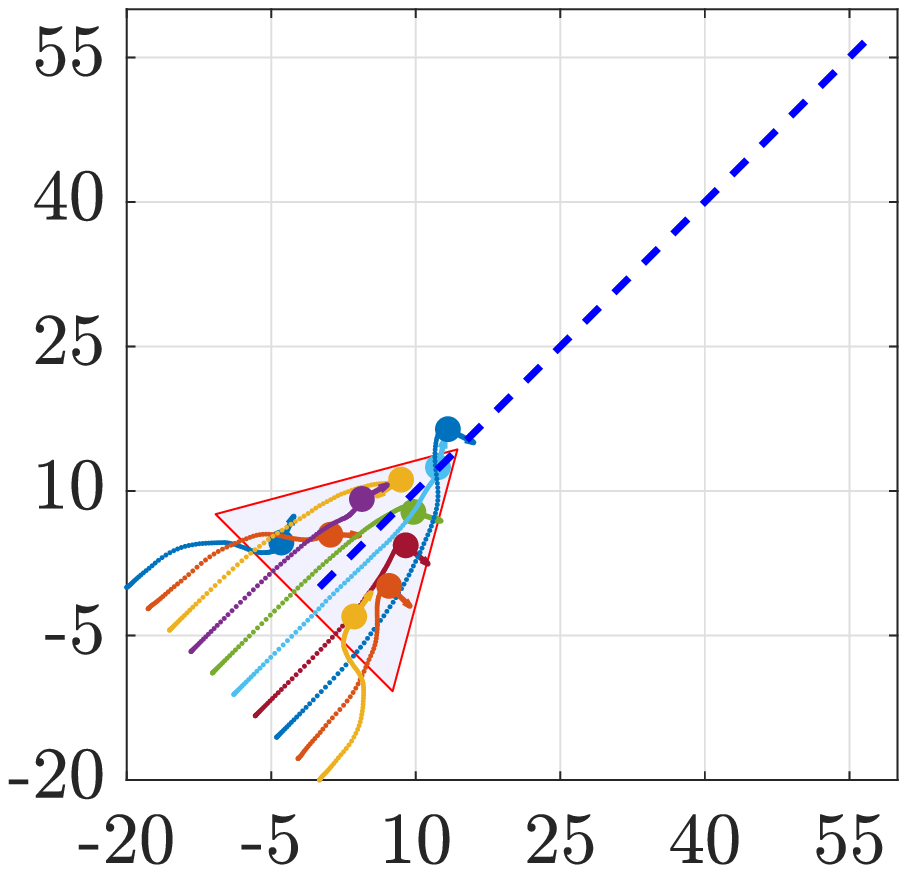}
\vspace{-2.5em}\caption{$t=9(s)$}
\label{fig:moving_triangle_avoid_9}
\end{subfigure}
\begin{subfigure}[b]{0.52\columnwidth}
\includegraphics[width=\columnwidth]{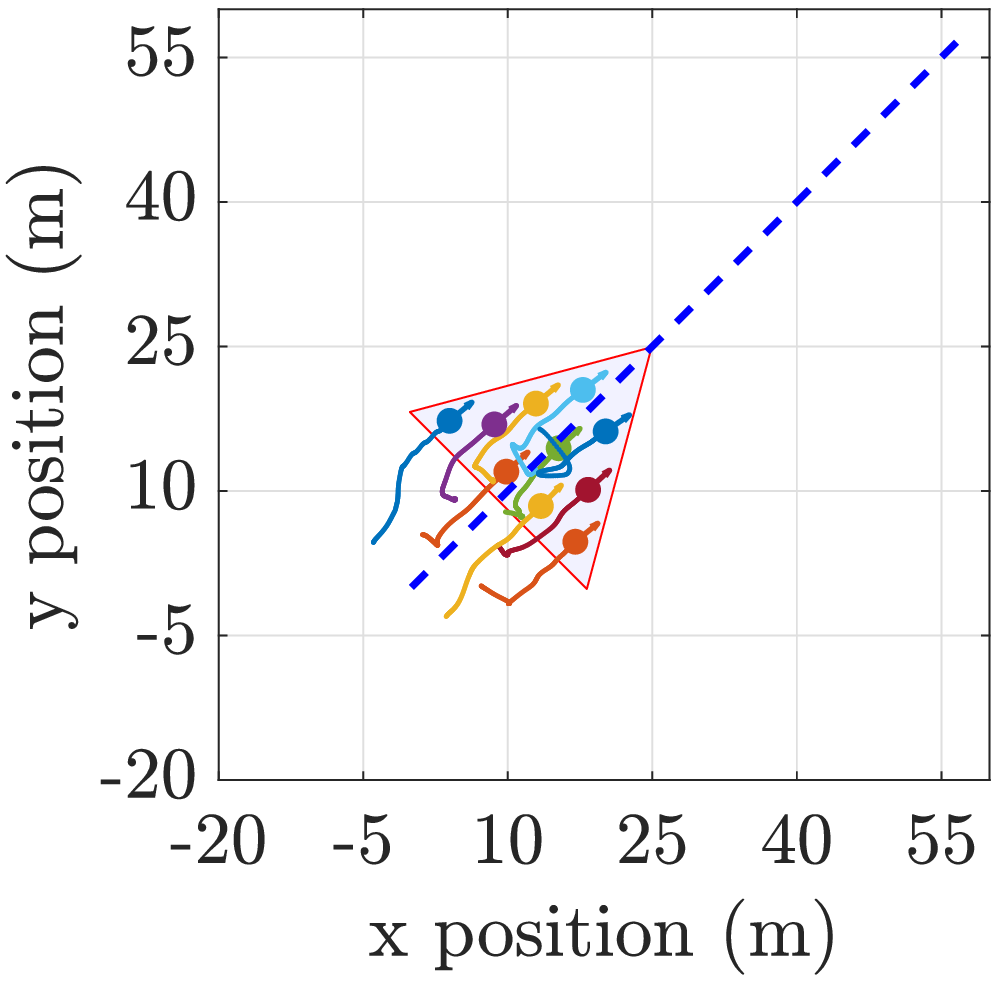}
\vspace{-1.5em}\caption{$t=24(s)$}
\label{fig:moving_triangle_avoid_24}
\end{subfigure}
\hspace{-1.9em}
\begin{subfigure}[b]{0.52\columnwidth}
\includegraphics[width=\columnwidth]{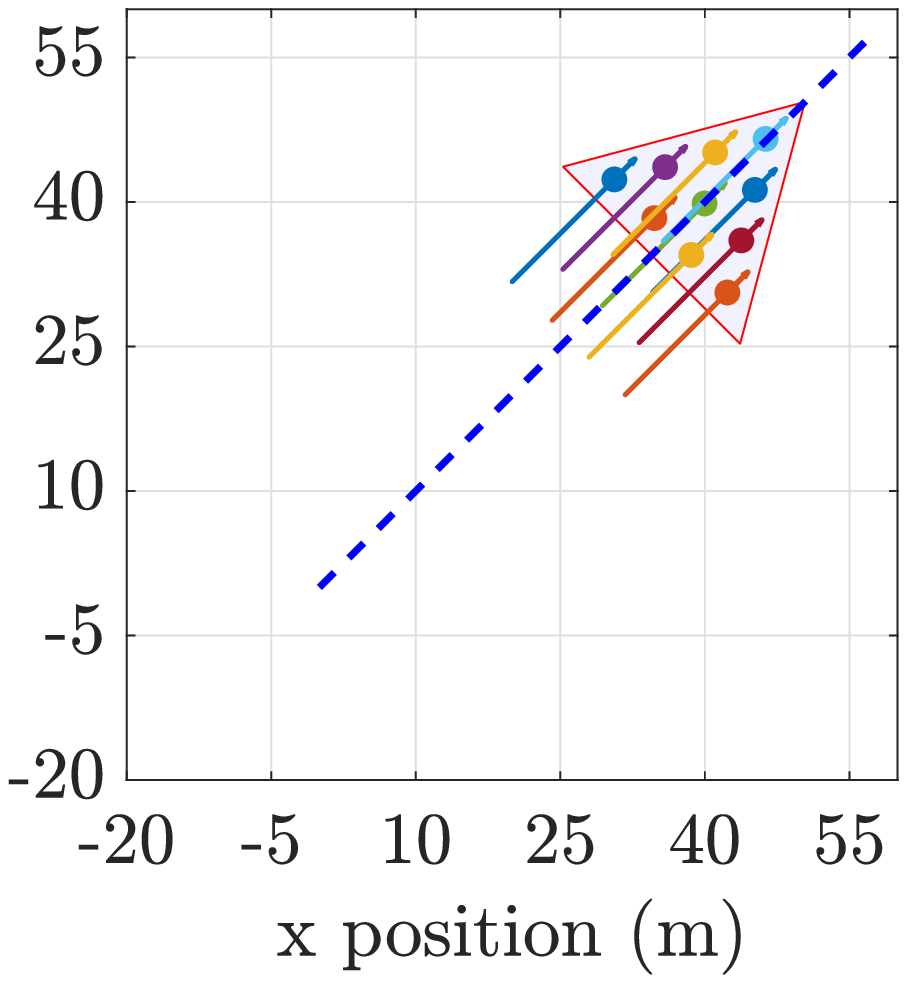}
\vspace{-1.5em}\caption{$t=60(s)$}
\label{fig:moving_triangle_avoid_60}
\end{subfigure}
\caption{Vehicles covering and following a moving triangular domain, when $N=10$, $c_{r}=2\,\text{(m)}$, $v_{max}=10\,\text{(m/s)}$, $u_{max}=3 \,\text{(m/s$^{2})$}$, $t_{safety}=5\,\text{(s)}$, $\aI=1\,\text{(m/s$^{2})$}$, $\ah=2\,\text{(m/s$^{2})$}$, $\av=0.2\,\text{(m/s$^{2})$}$, $C_{al}=0.2\,\text{(m/s$^{2})$}$, $l_{al}=7.79\,\text{(m)}$, $v_{d}=\left(\frac{\sqrt{2}}{2},\frac{\sqrt{2}}{2}\right)\,\text{(m/s)}$, $A=292.28\,\text{(m$^2$)}$ and $r_{d}=\sqrt{\frac{A}{N}}=5.4\,\text{(m)}$. Vehicles start in linear formation.}
\label{fig:moving_triangle}
\end{figure}

The objective of coverage control problems is to deploy agents to a possibly moving target area such that they can achieve an optimal sensing of the domain of interest. 
A common solution is through minimizing a coverage functional involving a Voronoi tessellation and the locations of vehicles within the tessellation \cite{Cortes_etal2004,Gao_etal2008}. 
This is a high-dimensional optimization problem which needs to be solved in real time.  
In our approach we achieve coverage\footnote{Alternative terminologies are balanced or anti-consensus configurations} through swarming by artificial potentials \cite{LeonardFiorelli2001,Sepulchre_etal2007, Chung_etal2018}. 
In a related problem, artificial potentials have been used for containment of follower agents within the convex hull of leaders \cite{RenCao2011,Cao_etal2012}.

Reachability analysis has been studied and used extensively in the past several decades as a tool for providing guarantees on performance and safety of dynamical systems \cite{Althoff2011,Frehse2011,Chen2013}, as well as controller synthesis in many cases.
In particular, Hamilton-Jacobi (HJ) reachability \cite{Yang2013,Chen2018} has seen success in collision avoidance \cite{Alam11,Chen2016}, air traffic management \cite{Margellos13,Chen2018a}, and emergency landing \cite{Akametalu2018}. 
HJ reachability analysis is based on dynamic programming, and involves solving an HJ partial differential equation (PDE) to compute a backward reachable set (BRS) representing states from which danger is inevitable. 
By using the derived optimal controller on the boundary of the BRS, safety can be guaranteed despite the worst-case actions of another agent.

In this paper we develop a new approach to self-collective coordination of autonomous agents that aim to reach and cover a moving target domain.
We consider two types of planar vehicle dynamics which differ in nature by the allowed control actions. 
The first, double integrator dynamics, allows the controller to specify the $x$ and $y$ accelerations at any time; among others, it is a simplified model for quadrotors. 
The second, planar fixed-wing aircraft dynamics, in which the vehicle controls the acceleration and turn rate, is a natural model for cars, bicycles or planes.

Our approach aims to enable the following: i) reaching and spreading over a target domain without having set \textit{a priori} the coverage configuration and the final state of each vehicle, ii) the use of a distributed control policy from which self-organization and intelligence emerges at the group level, and (iii) guarantee of collision-avoidance throughout the coordination process.

In this aim, we consider a control policy that includes both a coverage and a safety controller. 
The coverage controller brings vehicles inside a target domain, spreads them over the target domain, and aligns vehicle and domain velocities with each other. 
On the other hand, the safety controller guarantees collision avoidance of vehicles.
A simulation of the emergent behaviour resulting from our controller is shown in Figure \ref{fig:moving_triangle}, where $N=10$ vehicles move to cover a moving triangular domain.

The proposed coverage controller uses two types of artificial potentials and velocity alignment terms resembling the Cucker-Smale model with single leader \cite{li2010cucker,ha2014flocking}. 
One artificial potential is for inter-individual forces which are designed to achieve a certain desired inter-vehicle spacing as in \cite{LeonardFiorelli2001}. 
Such controller enables emergent self-collective behaviour of the vehicles, similar to the highly coordinated motions observed in biological groups such as flocks of birds and schools of fish  \cite{Camazine_etal2003}. 
The other artificial potential is used for vehicle-target forces by which vehicles reach the target and cover it. 
The Cucker-Smale terms promote the vehicles to match the velocity of the target domain, which acts as a leader.

We emphasize that the proposed coverage controller, which also drives vehicles inside the target domain, is done through agent swarming; there is no leader and no order among the agents.
This means that the controller does not rely on the well functioning of each individual agent.
Such self-collective and cooperative behaviour is present in systems of interacting agents in the physics and biology literature \cite{Vicsek_etal1995,Couzin_etal2002,DOrsogna_etal2006,CuckerSmale2007,FetecauGuo2012}. An agent search and target-locating algorithm based on a swarming model was studied in \cite{Liu_etal2010}.

Unlike first-order models, where agents directly control their velocities, our vehicle models are second-order: agents are implicitly or explicitly controlled through their acceleration. 
In addition, We set \textit{a priori} bounds on the control forces, making our controller more realistic than previous approaches, in which infinite forces may be needed to guarantee collision avoidance \cite{LeonardFiorelli2001,HusseinStipanovic2007}.



The safety controller for vehicles with double integrator dynamics is derived from HJ reachability analysis. 
Instead of numerically solving an associated Hamilton-Jacobi-Isaacs (HJI) PDE, we derive the analytical solution to the PDE to eliminate numerical errors and the need to specify computation bounds. 
While multi-vehicle collision avoidance is in general intractable, incorporating pairwise collision avoidance  drastically reduced the collision rate.

The paper is organized as follows. Section \ref{sect:HJ-reach} presents some background on Hamilton-Jacobi reachability. In Section \ref{sect:static} we study the safe coverage problem for static domains with double integrator dynamics. In Section \ref{sect:moving} we generalize the study to moving domains following inertial trajectories. In Section \ref{sect:fixed-wing} we formulate a control algorithm for coverage of moving domains  with fixed-wing aircraft dynamics. 
Finally, we make concluding remarks and discuss open problems and potential future directions of research.


\section{Background: Hamilton-Jacobi Reachability}
\label{sect:HJ-reach}
We review here some basic Hamilton-Jacobi reachability theory, which will be used in the paper to address pairwise collision avoidance.

Consider the two-player differential game described by the joint system
\begin{align}
\begin{aligned}
\label{eq:syst-gen}
\dot{z}\left(t\right) & =f\left(z\left(t\right),u\left(t\right),d\left(t\right)\right),\\
z\left(0\right) & =x,
\end{aligned}
\end{align}
\noindent where $z\in\mathbb{R}^{n}$ is the joint state of the players, $u\in\mathcal{U}$ is the control input of Player 1 (hereafter referred to as ``control'') and $d\in\mathcal{D}$ is the control input of Player 2 (hereafter referred to as ``disturbance'') .

We assume $f:\mathbb{R}^{n}\times\mathcal{U}\times\mathcal{D}\rightarrow\mathbb{R}^{n}$ is uniformly continuous, bounded, and Lipschitz continuous in $z$ for fixed $u$ and $d$, and $u\left(\cdot\right)\in\mathcal{U}$, $d\left(\cdot\right)\in\mathcal{D}$ are measurable functions. Under these assumptions we can guarantee the dynamical system \eqref{eq:syst-gen} has a unique solution.

In this differential game, the goal of player 2 (the disturbance) is to drive the system into some target set using only non-anticipative strategies \cite{Yang2013}, while player 1 (the control) aims to drive the system away from it.

We introduce the \textit{time-to-reach} problem as follows.

\textit{\textbf{(Time-to-reach) } Find the time to reach a target set $\domDanger$ while avoiding the obstacle $\domSafe$ from any initial state $x$, in a scenario where  player 1 maximizes the time, while player 2 minimizes the time.
Player 2 is restricted to using \textit{non-anticipative} strategies,  with knowledge of player 1's current and past decisions. 
Such a time is denoted by $\phi\left(x\right)$.}

Following \cite{Yang2013}, given $u\left(\cdot\right)$ and $d\left(\cdot\right)$, the time to reach a closed target set $\domDanger$ with compact boundary, while avoiding the obstacle $\domSafe$, is defined as
\[
T_{x}\left[u,d\right]=\min\left\{ t|\,z\left(t\right)\in\domDanger \text{ and } z\left(s\right)\notin\domSafe, \forall s\in\left[0,t\right] \right\}.
\]
Then, the \textit{Time-to-reach} problem reduces to finding:
\[
\phi\left(x\right)=\underset{\theta\in\Theta}{\min}\,\underset{u\in\mathcal{U}}{\max}\,T_{x}\left[u,\theta\left[u\right]\right],
\]
where $\Theta$ represents the set of non-anticipative strategies. The collection of all the states that are reachable in a finite time is the capturability set $\mathcal{R}^{*}=\left\{ x\in\mathbb{R}^{n}|\; \phi\left(x\right)<+\infty\right\} $.

Applying the dynamic programming principle, as done in \cite{bardi1989boundary}, one can obtain $\phi$ as the viscosity solution of the following stationary HJ PDE:
\begin{equation}
\label{eq:HJPDE}
\begin{aligned}
\underset{u\in\mathcal{U}}{\min}\,\underset{d\in\mathcal{D}}{\max}\,\left\{ -\nabla\phi\left(z\right)\cdot f\left(z,u,d\right)-1\right\}=0\text{ in }&\mathcal{R}^{*}\backslash\left(\domDanger\cup\domSafe\right), \\
\phi\left(z\right)=0\text{ on }\domDanger,\qquad\phi\left(z\right)=\infty\text{ on  }\domSafe.& 
\end{aligned}
\end{equation}
In applications, this PDE is typically solved using  finite difference methods such as the Lax-Friedrichs method \cite{Yang2013}. Also, from the solution $\phi\left(x\right)$ one can obtain the control input for optimal avoidance as:
\begin{align}\label{eq:optimalControl}
u^{*}\left(z\right)=\arg\underset{u\in\mathcal{U}}{\min}\,\underset{d\in\mathcal{D}}{\max}\left\{-\nabla\phi\left(z\right)\cdot f\left(z,u,d\right)-1\right\}.
\end{align}

%% file: static_domain.tex
\section{Coverage of a static domain}
\label{sect:static}

\subsection{Problem formulation}
\label{subsect:problem}
We consider a group of $N$ vehicles, denoted by $\Qi$, $i=1,\dots,N$, with the double integrator dynamics given by
\begin{equation} \label{eq:problemDynamic}
\begin{split}
\dot{p}_{i}=v_{i}, \quad \dot{v}_{i}=u_{i}; \qquad \left\Vert v_{i} \right\Vert \leq v_{max}, \, \left\Vert u_{i} \right\Vert \leq u_{max}.
\end{split}
\end{equation}
Here, $p_{i}=\left(\pix,\piy\right)$ and $v_{i}=\left(\vix,\viy\right)$ are the position and velocity of  $\Qi$ respectively, and  $u_{i}=\left(\uix,\uiy\right)$ is the control force applied to this vehicle.

Given a predefined collision radius $c_{r}$, a vehicle is considered safe if there is no other vehicle within distance $c_{r}$ to it, i.e., if
\begin{equation}
    \centering
    \left\Vert p_{i}-p_{j} \right\Vert>c_{r}, \qquad  \text{ for any } j \neq i.
    \label{eq:safety}
\end{equation}

In this paper we are interested in certain configurations of the agents in the domain $\Omega$. Specifically, we set the following definitions.
\begin{defn}[$r$-Subcover] \label{defn:subcover}
A group of agents is an \textbf{$r$-subcover} for a compact domain $\Omega\subseteq\RR$ if:
\begin{enumerate}
\item The distance between any two vehicles is at least $r$.
\item The signed distance from any vehicle to $\Omega$ is less than equal to $-\frac{r}{2}$.
\end{enumerate}
\end{defn}
\begin{defn}[$r$-Cover]
\label{defn:cover} An $r$-subcover for $\Omega$ is an \textbf{$r$-cover} for $\Omega$ if its size is maximal (i.e., no larger number of agents can be an $r$-subcover for $\Omega$). 
\end{defn}

The $r$-subcover definition is closely related to the packing problem for circular objects of radius $\frac{r}{2}$ in a container with shape $\Omega$. Having an $r$-cover implies the container is full and there is no room for more of such objects.
\medskip

The following safe domain coverage problem is of main interest to our work in this chapter.

\textbf{Safe-domain-coverage by vehicles with double integrator dynamics:} \textit{Consider a compact domain $\Omega$ in the plane and $N$ vehicles each with dynamics described by \eqref{eq:problemDynamic}, starting from safe initial conditions. Find the maximal $r>0$ and a control policy that leads to a stable steady state which is an $r$-cover for $\Omega$, while satisfying the safety condition \eqref{eq:safety} at any time.}
\medskip

The controller we design and present below has two components: a coverage controller and a safety controller, the latter being based on Hamilton-Jacobi reachability. We will present them separately.


\subsection{Coverage controller}
\label{subsect:coverage-controller}
Define $p_{ij}:=p_{i}-p_{j}$, and denote by $P_{\partial\Omega}\left(p_{i}\right)$ the closest point of $\partial \Omega$ to $p_{i}$ (i.e., the projection of $p_i$ on $\partial \Omega$). Also, define $h_{i}:=p_{i}-P_{\partial\Omega}\left(p_{i}\right)$, and denote by $\left[\left[ h_{i}\right] \right]$ the signed distance of $p_i$ from  $\partial\Omega$ -- see Figure \ref{fig:controlForce}. 

The proposed control force is given by:
\begin{equation}\label{eq:controlExplicit}
u_{i}=\!-\!\!\sum_{j\neq i}^{N} f_{I}\left(\left\Vert p_{ij}\right\Vert \right)\frac{p_{ij}}{\left\Vert p_{ij}\right\Vert }-f_{h}\left(\left[\left[h_{i}\right] \right]\right)\frac{h_{i}}{\left[\left[h_{i}\right] \right] } - a_v v_i,
\end{equation}
\noindent 
where the three terms in the right-hand-side represent inter-vehicle, vehicle-domain, and braking forces, respectively. We assume each vehicle is able to measure its distance to the target domain, its speed, as well as its position relative to the other vehicles. In \eqref{eq:controlExplicit}, $a_v$ is a fixed positive constant. 

Figure \ref{fig:controlForce} illustrates the control forces for two generic vehicles located at $p_i$ and $p_j$. Shown there are the unit vectors in the directions of the inter-vehicle and vehicle-domain forces (yellow and blue arrows, respectively), along with the resultant that gives the overall control force (red arrows). Note that due to the nonsmoothness of the boundary, different points may have different types of projections: $p_i$ projects on the foot of the perpendicular to $\partial \Omega$, while $p_j$ projects on a corner point of $\partial \Omega$.

\begin{figure}
\centering
    \begin{subfigure}[b]{0.75\linewidth}        \includegraphics[width=\linewidth]{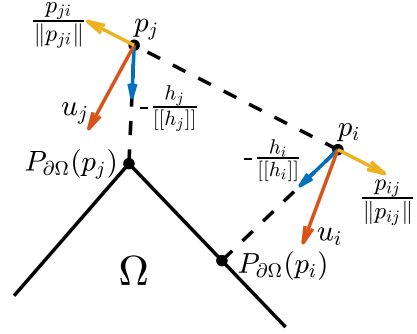}\hspace{-0.1cm}
    \end{subfigure}
    \vspace{-1.6em}
\caption{Illustration of control forces acting on two vehicles located at $p_i$ and $p_j$.} 
\label{fig:controlForce}
\end{figure}

 Figure \ref{fig:fINfh} shows the specific forms of the functions $f_{I}$ and $f_{h}$ that we consider in this paper. Note that $f_{I}(r)$ is negative for $r<\rd$, and zero otherwise. This means that for two vehicles within distance $0<r<\rd$ from each other, their inter-vehicle interactions are repulsive, while two vehicles at distance larger than $\rd$ apart do not interact at all. The vehicle-domain force $f_{h}(r)$ is zero for $r<-\frac{\rd}{2}$, and positive for $r>-\frac{\rd}{2}$.  For a vehicle $i$ outside the target domain, i.e., with $\left[\left[h_{i}\right] \right] >0$, this results in an attractive interaction force toward $\partial \Omega$. On the other hand, for a vehicle inside the domain, where $\left[\left[h_{i}\right] \right] <0$, one distinguishes two cases: i) the vehicle is within distance $\frac{\rd}{2}$ to the boundary, in which case it experiences a repulsive force from it, or ii) the vehicle is more than distance $\frac{\rd}{2}$ from the boundary, in which case it does not interact with the boundary at all. 

\begin{lem}
\label{lemma:potential}
The inter-vehicle and vehicle-domain forces are conservative.
\end{lem}
\begin{proof}
 Define the potentials: 
 \begin{equation*}
   V_{I}\left(p_{ij}\right)={\int_{r_{d}}^{\left\Vert p_{ij}\right\Vert }}f_{I}\left(s\right)ds, \qquad   V_{h}\left(p_{i}\right)={\int_{-\frac{r_{d}}{2}}^{\left[\left[ h_{i}\right]\right] }} f_{h} \left(s\right) ds.
 \end{equation*}
Then,
\begin{align*}
\nabla_{i} V_{h}\left(p_{i}\right)&=f_{h}\left(\left[\left[h_{i}\right] \right] \right)\nabla(\left[\left[h_{i}\right] \right])=f_{h}\left(\left[\left[h_{i}\right] \right] \right)\frac{h_{i}}{\left[\left[h_{i}\right] \right] },
\end{align*}
where we have used the identity $\nabla(\left[\left[h_{i}\right]\right]) = \frac{h_{i}}{\left[\left[h_{i}\right]\right]}$ (see Theorem 5.1(iii) in \cite{DelfourZolesio2001}). Similarly, the inter-vehicle force is the negative gradient of the potential $V_I$.
\end{proof}

The potentials $V_{I}$ and $V_{h}$ are shown in Figure \ref{fig:potentials}. Their explicit expressions are given by:
\begin{equation}
\label{eq:VI}
V_I(x) = 
\begin{cases}
\frac{\aI}{2}(\|x\|-\rd)^2 &\quad  \text{ for } \|x\|<\rd,\\
0 & \quad \text{ for } \|x\|\geq \rd,
\end{cases}
\end{equation}
and
\small
\begin{equation}
\label{eq:Vh}
V_h(x) = 
\begin{cases}
0 & \text{for } [[x - P_{\partial \Omega} x]]\leq -\frac{\rd}{2}, \\
\frac{\ah}{2} ([[x - P_{\partial \Omega} x]] + \frac{\rd}{2})^2 & \text{for } [[x - P_{\partial \Omega} x]]>-\frac{\rd}{2},
\end{cases}
\end{equation}
\normalsize
where $\aI>0$ is the slope of the function $f_I$ on $[0,\rd]$ and $\ah>0$ is the slope of the function $f_h$ on $\left[-\frac{\rd}{2},\infty\right)$. Note that as $a_{I}$ and $a_{h}$ are positive, both potentials $V_{I}$ and $V_h$ are non-negative.

By Lemma \ref{lemma:potential}, the control given in Equation \eqref{eq:controlExplicit} becomes
\begin{equation}\label{eq:controlPotential}
u_{i}={\sum_{j\neq i}^{N}}-\nabla_{i}V_{I}\left(p_{ij}\right)-\nabla_{i}V_{h}\left(p_{i}\right)- a_v v_i.
\end{equation}

\begin{figure}
\centering
    \begin{subfigure}[b]{0.45\columnwidth}
        \includegraphics[width=\columnwidth]{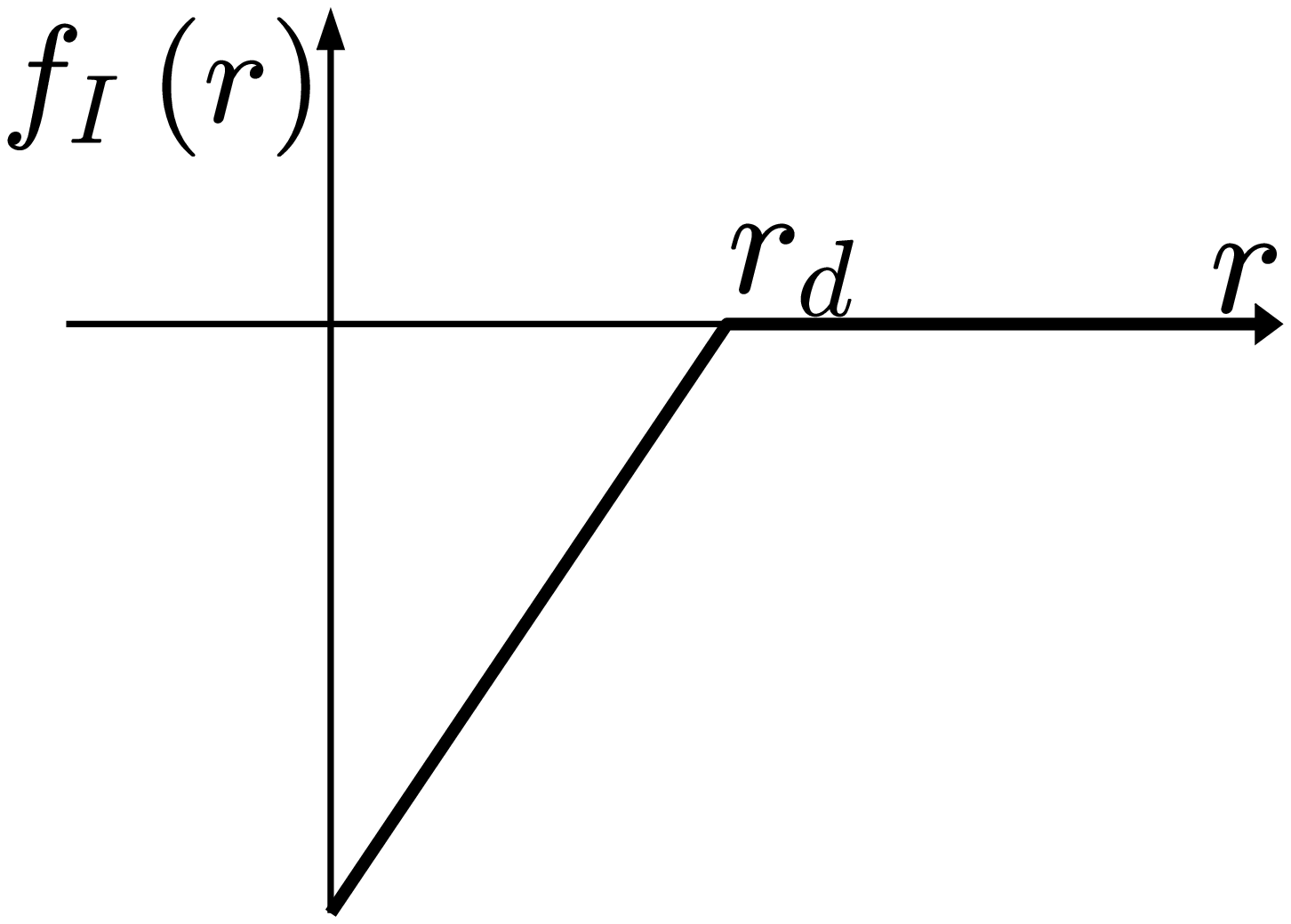}
    \end{subfigure}
    \begin{subfigure}[b]{0.45\columnwidth}
        \includegraphics[width=\columnwidth]{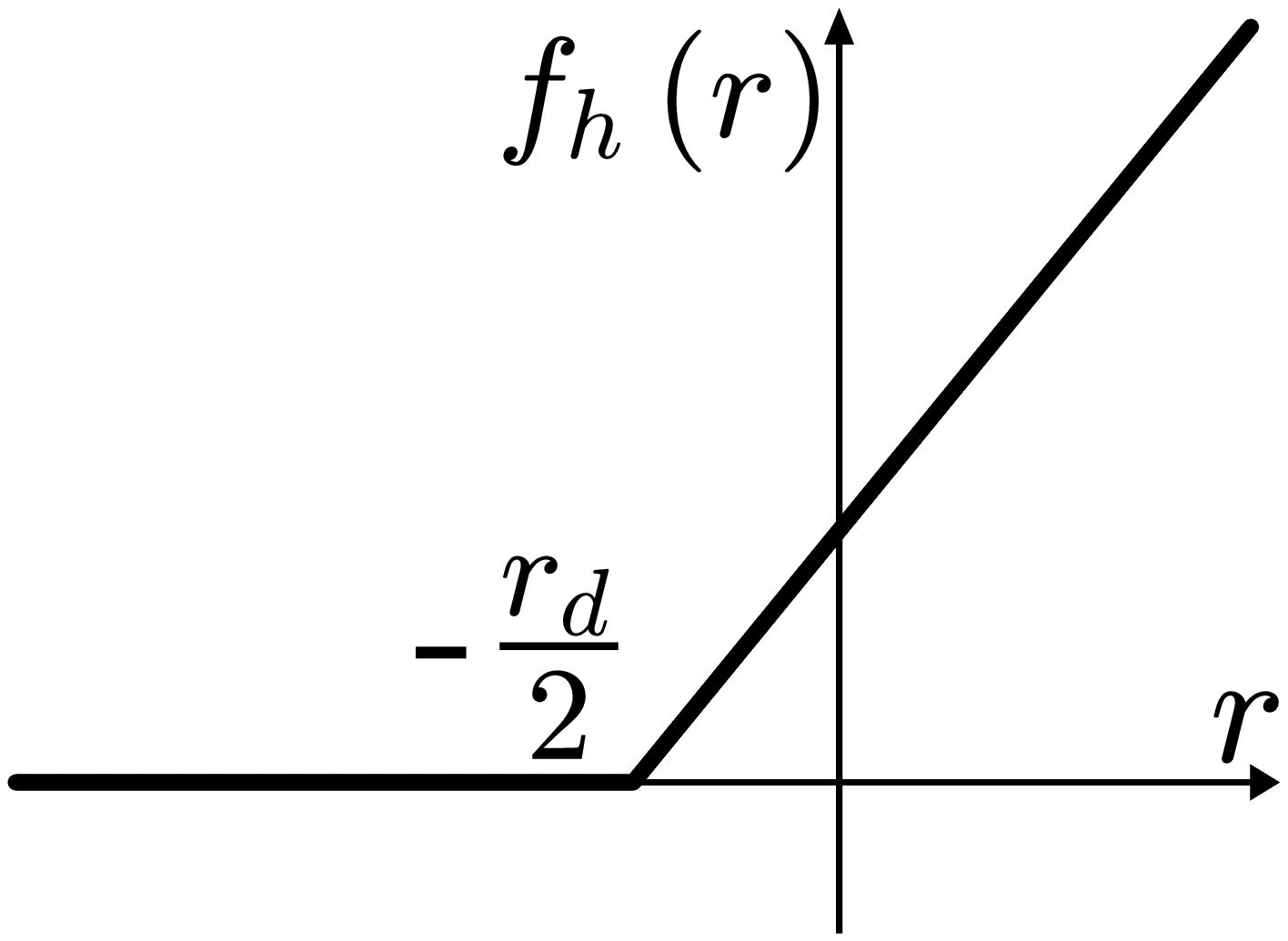}
    \end{subfigure}
\caption{Inter-vehicle and vehicle-domain control forces.} 
\label{fig:fINfh}
\end{figure}

\begin{figure}
\centering
    \begin{subfigure}[b]{0.45\columnwidth}
        \includegraphics[width=\columnwidth]{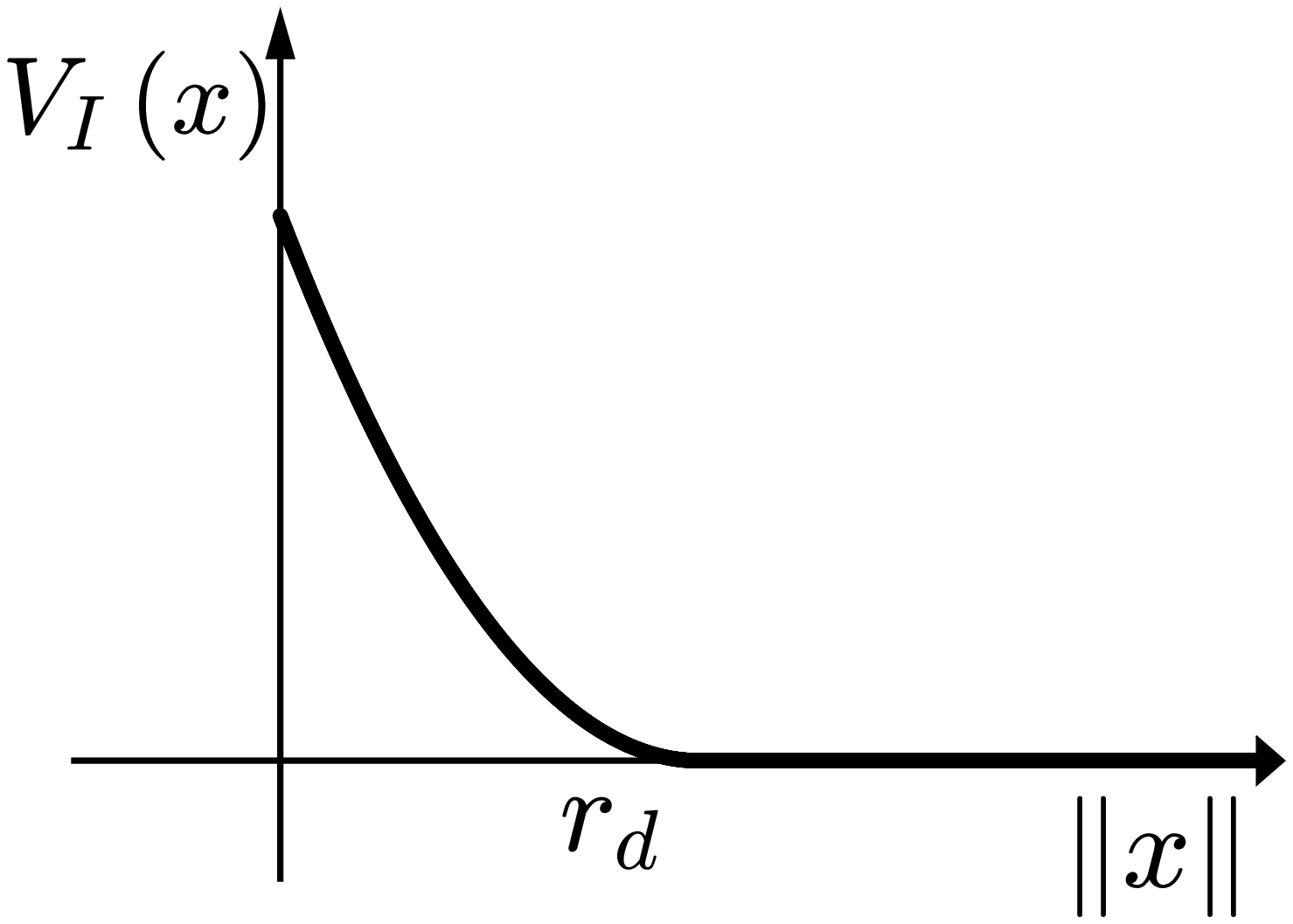}
    \end{subfigure}
    \begin{subfigure}[b]{0.45\columnwidth}
        \includegraphics[width=\columnwidth]{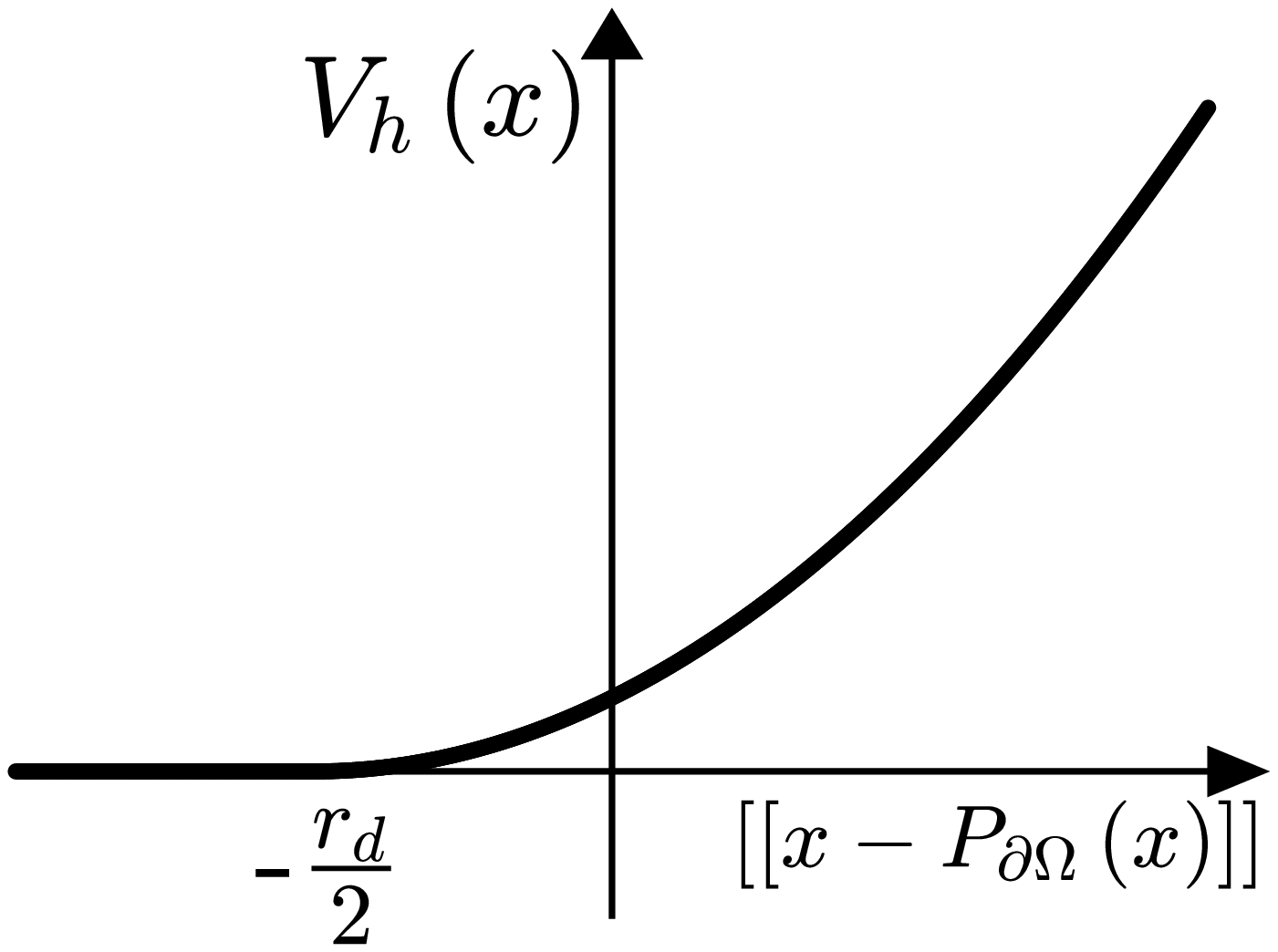}
    \end{subfigure}
\caption{Inter-vehicle and vehicle-domain potentials.} 
\label{fig:potentials}
\end{figure}

\medskip
\textbf{Asymptotic behaviour of the controlled system.} Consider the following candidate for a Lyapunov function, consisting in kinetic plus (artificial) potential energy:
\begin{equation}
\label{eq:Lyapunov}
\Phi=\frac{1}{2}{\sum_{i=1}^{N}}\Bigl(\dot{p}_{i}\cdot\dot{p}_{i}+{\sum_{j\neq i}^{N}}V_{I}\left(p_{ij}\right)+2V_{h}\left(p_{i}\right)\Bigr).
\end{equation}
Note that each term in $\Phi$ is non-negative, and $\Phi$ reaches its absolute minimum value when the vehicles are totally stopped. Also, at the global minimum $\Phi=0$, the equilibrium configuration is an $\rd$-subcover of $\Omega$; in particular, all vehicles are inside the target domain.


The time derivative of $\Phi$ can be calculated as:
\begin{align*}
\dot{\Phi} & ={\sum_{i=1}^{N}}\dot{p}_{i}\cdot\Bigl(u_{i}+{\sum_{j\neq i}^{N}}\nabla_{i}V_{I}\left(p_{ij}\right)+\nabla_{i}V_{h}\left(p_{i}\right)\Bigr)\\
 &=-{\sum_{i=1}^{N}} a_v \| v_{i}\|^2,
\end{align*}
where we used the dynamics \eqref{eq:problemDynamic} and equation \eqref{eq:controlPotential}. Note that $\dot{\Phi}$ is negative semidefinite and equal to zero if and only if $v_{i}=0$ for all $i$ (i.e., all vehicles are at equilibrium). 

We first show that the group of vehicles remains within a compact set through time evolution. The key idea is that the vehicle-domain potential $V_h$ is confining the vehicles, and keeps them as a group \cite{Olfati2006}.
\begin{prop} 
\label{prop:cohesive}
Solutions of \eqref{eq:problemDynamic}, with control law given by \eqref{eq:controlPotential} remain cohesive through time, i.e., there exists an $R>0$ such that $\|\pos_i(t)\| \leq R$, for all $i$ and $t \geq 0$.
\end{prop}
\begin{proof} Using that the kinetic energy and the potential $V_I$ are non-negative, and $\Phi$ given by \eqref{eq:Lyapunov} is non-increasing, we have:
\[
\sum_{i=1}^{N}V_{h}\left(\pos_{i}(t)\right) \leq \Phi(t) \leq \Phi(0).
\]
To show the boundedness of $\pos_i$ we only need to consider the case $\pos_i \notin \Omega$, as otherwise the vehicles are inside the compact set $\Omega$. Using the expression \eqref{eq:Vh} for $V_h$, we then find:
\[
\frac{\ah}{2} \sum_{i=1}^{N} \left( \|\pos_i(t) - P_{_{\partial \Omega}} \pos_i(t)\| + \frac{\rd}{2} \right)^2 \leq \Phi(0).
\]
This shows that the distances from $\pos_i(t)$ to the domain $\Omega$ remain bounded for all $t \geq 0$ by $\sqrt{\frac{2 \Phi(0)}{\ah}}$ when $\pos_i \notin \Omega$. 
\end{proof}

\begin{rmk}
\label{rem:LaSalle}
From LaSalle Invariance Principle we can conclude that the controlled system approaches asymptotically an equilibrium configuration. By the expressions \eqref{eq:controlPotential} of the control force and \eqref{eq:Lyapunov} of the Lyapunov function, these are equilibria that are critical points of the artificial potential energy $\frac{1}{2}{\sum_{i=1}^{N}}\Bigl({\sum_{j\neq i}^{N}}V_{I}\left(\pos_{ij}\right)+2V_{h}\left(\pos_{i}\right)\Bigr)$. We expect that any critical point other than the local minima (e.g., saddles or local maxima) are unstable \cite{Olfati2006}, and hence, almost every solution of the system will approach asymptotically a local minimum of the potential energy.
\end{rmk}

For certain simple setups (e.g., a square number of vehicles in a square domain -- see Figure \ref{fig:squareAvoid50}, or a triangular number of vehicles in a triangular domain), the $r_{d}$-covers are isolated equilibria. Hence, together with the fact that such equilibria are global minimizers for $\Phi$, their local asymptotic stability can be inferred. The formal result is given by the following proposition.

\begin{prop}
\label{prop:stability}
Consider a group of $N$ vehicles with dynamics defined by \eqref{eq:problemDynamic}, and the control law given by \eqref{eq:controlPotential}. Let the equilibrium of interest be of the form $\dot{p_{i}}=0$,  $\left\Vert p_{ij}\right\Vert \geq r_{d}$  and $\left[\left[h_{i}\right]\right]\leq -\frac{r_{d}}{2}$ for  $i,j=1,\cdots,N$ (see Definitions \ref{defn:subcover} and \ref{defn:cover}), and assume that this equilibrium configuration is isolated. Also assume that there is a neighborhood about the equilibrium in which the control law remains smooth. Then, the equilibrium is a global minimum of the sum of the artificial potentials and is locally asymptotically stable.
\end{prop}
\begin{proof}
The proof follows from LaSalle invariance principle and the arguments made above.
\end{proof}

Choosing an adequate $r_{d}$ when solving the \textit{safe-domain-coverage} problem leads to a nonlinear optimization problem (see \cite{lopez2011heuristic}), which in general can be quite difficult. We set the value of this parameter based on the assumption that any vehicle is covering roughly the same square area, i.e.,
\begin{equation}
\label{eq:rd}
    r_d = \sqrt{\frac{Area\left(\Omega\right)}{N}}.
\end{equation}
Note that \eqref{eq:rd} gives the exact maximal radius when  both the number of vehicles and the domain are square. The numerical experiments presented in this paper, which also involve target domains in the shape of a triangle or an arrowhead, show that \eqref{eq:rd} leads indeed to the desired covers.


\subsection{Collision avoidance}
\label{subsect:collision-controller}

An important component of our study is the guarantee that vehicles do not collide through the time evolution. 
For small initial energies, collision avoidance can be shown directly. For general cases, we introduce a safety controller based on HJ reachability analysis.
\medskip

\textbf{Small initial energy.} The following results hold for initial data with small energy $\Phi$.
\begin{prop}
\label{prop:collision-avoid}
Consider a target domain $\Omega$ and a group of $N$ vehicles with dynamics defined by \eqref{eq:problemDynamic} and \eqref{eq:controlPotential}. Assume the energy $\Phi(0)$ of the initial configuration satisfies
\[
\Phi(0) < \int_{\rd}^{c_{r}} f_{I}(s) \, ds = \frac{\aI}{2} (c_{r}-\rd)^2.
\]
Then, no vehicle collision can occur for all $t \geq 0$.
\end{prop}
\begin{proof}
Suppose by contradiction that there is a time $\tc$ at which vehicles $k$ and $l$ are at collision radius from each other, that is, $\|\pos_k(\tc)-\pos_l(\tc)\|= c_{r}$. Given that $V_I$ is non-negative, the inter-vehicle potential energy at the collision time can be bounded below as:
\[
\frac{1}{2}{\sum_{i=1}^{N}}{\sum_{j\neq i}^{N}}V_{I}\left(\pos_{ij}(\tc)\right) \geq V_I (\pos_k(\tc) - \pos_l(\tc)) = \int_{\rd}^{c_{r}} f_{I}(s) \, ds.
\]
On the other hand, using that the kinetic energy and the potential $V_h$ are non-negative, we have:
\[
\frac{1}{2}{\sum_{i=1}^{N}}{\sum_{j\neq i}^{N}}V_{I}\left(\pos_{ij}(\tc)\right) \leq \Phi(\tc) \leq \Phi(0).
\]
By combining the two sets of inequalities above one finds $\Phi(0) \geq \int_{\rd}^{c_{r}} f_{I}(s) \,ds$, which contradicts the assumption on the initial energy $\Phi(0)$.
\end{proof}

The result above can be generalized as follows.
\begin{prop}
\label{prop:collision-avoid-k}
Consider a target domain $\Omega$ and a group of $N$ vehicles with dynamics defined by \eqref{eq:problemDynamic} and \eqref{eq:controlPotential}. Assume the energy $\Phi(0)$ of the initial configuration satisfies
\[
\Phi(0) < \left(k+1\right)\int_{r_{d}}^{c_{r}}f_{I}(s) \, ds,
\]
for some $k\in\mathbb{Z}_{+}.$ Then, at most $k$ distinct pairs of vehicles could be possibly unsafe ($k = 0$ guarantees a safe motion) for all $t \geq 0$.
\end{prop}
\begin{proof}  Assume by contradiction that $k+1$ pairs of vehicles are unsafe at time $\tc$, i.e., their relative distances are less than or equal to $c_{r}$ at $\tc$. Then, on one hand, following the argument in Proposition \ref{prop:collision-avoid}, we have:
\[
\frac{1}{2}{\sum_{i=1}^{N}}{\sum_{j\neq i}^{N}}V_{I}\left(\pos_{ij}(\tc)\right) \geq (k+1) \int_{\rd}^{c_{r}} f_{I}(s) \, ds,
\]
where we use the fact that $V_{I}(\pos_{ij})$ is non-negative and non-increasing. 

On the other hand,
\[
\frac{1}{2}{\sum_{i=1}^{N}}{\sum_{j\neq i}^{N}}V_{I}\left(\pos_{ij}(\tc)\right) \leq \Phi(\tc) \leq \Phi(0),
\]
leading to a contradiction.
\end{proof}

Note that the two last results assume that the control law \eqref{eq:controlPotential} is applied as it is, that is, it does not take into account the input force constrains in \eqref{eq:problemDynamic}. The following HJ reachability analysis deals with the input force bounds to guarantee pairwise safety.

\medskip
\textbf{Collision avoidance via Hamilton-Jacobi theory.} To guarantee pairwise collision avoidance for general configurations, 
we design a safety controller based on HJ reachability analysis.   

Consider the dynamics between two vehicles $\Qi$, $\Qj$ defined in terms of their relative states
\begin{align*}
\pxr &=\pix-\pjx, \qquad \vxr =\vix-\vjx, \\
\pyr &=\piy-\pjy, \qquad \vxr =\vix-\vjx,
\end{align*}
\noindent where the vehicle $\Qi$ is the evader, located at the origin, and $\Qj$ is the pursuer, the latter being considered as the model disturbance. The relative dynamical system can be written as:
\begin{align}
\label{eq:relativeDynSys}
\pdxr &= \vxr, \qquad \vdxr = \uix-\ujx, \nonumber\\
\pdyr &= \vyr, \qquad \vdyr = \uiy-\ujy,
\end{align}
\noindent with $\left\Vert u_{i}\right\Vert$,$\left\Vert u_{j}\right\Vert \leq u_{max}$, where $u_i=\left(\uix,\uiy\right)$ and $u_j=\left(\ujx,\ujy\right)$ are the control inputs of the agents $\Qi$ and $\Qj$, respectively. From the perspective of agent $\Qi$, the control inputs of $\Qj$ are treated as worst-case disturbance.

System \eqref{eq:relativeDynSys} can be put in the general form \eqref{eq:syst-gen} from Section \ref{sect:HJ-reach}, with $z=(\pxr,\pyr,\vxr,\vyr)$, $u=\left(u_x,u_y\right):=\left(\uix,\uiy\right)$, $d=\left(d_x,d_y\right):=\left(\ujx,\ujy\right)$, and $f\left(z,u,d\right)$ being the right-hand-side of \eqref{eq:relativeDynSys}.

According to \eqref{eq:safety}, the unsafe states are described by the target set $\domDanger=\left\{ z:\pxr^{2}+\pyr^2\leq c_{r}^{2}\right\}$. For now, the obstacle set $\domSafe$ is the empty set as it is not needed until Section \ref{sec:fiked_wing_coverage_control}. Consider $\psi\left(z\right)$ as the time it takes for the solution of the dynamical system \eqref{eq:relativeDynSys}, with starting point $z$ in $\mathcal{R^{*}}\setminus\domDanger$, to reach $\domDanger$ when the disturbance and control inputs are optimal. As the two vehicles have the same capabilities we make the educated guess that the optimal non-anticipative strategy for the pursuer is to copy the evader accelerations, having so a zero relative acceleration. This implies that the relative velocity $v_r$ will remain constant through time.

If $p_r$ and $v_r$ are such that a collision can occur, there exist a collision point $c_p$, see Figure \ref{fig:geometric_argument}. This will be one of the intersection points of the line crossing through $p_r$ with direction parallel to $v_r$ and the circle of radius $c_r$ centered at the origin. To get the collision time we replace the coordinates of $c_p=\left(p_{r,x}+\psi\left(z\right)\vel_{r,x},\,p_{r,y}+\psi\left(z\right)\vel_{r,y}\right)$ into the canonical equation of the circle. Using this geometric argument one can show that this time is the minimum of the two solutions of the quadratic equation:
\begin{multline}\label{eq:quadraticTTR}
\left(\vxr^{2}+\vyr^{2}\right)\ttr^{2}\left(z\right)+2\left(\pxr\vxr+\pyr\vyr\right)\ttr\left(z\right)\\
+\left(\pxr^{2}+\pyr^{2}-c_{r}^{2}\right)=0.
\end{multline}
\begin{figure}
    \centering
    \includegraphics[width=0.65\columnwidth]{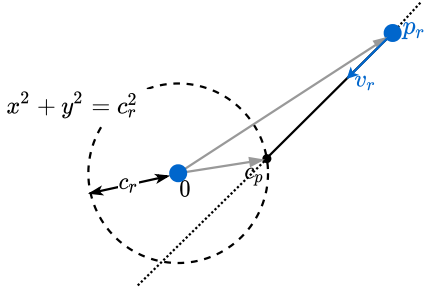}
    \caption{Geometric illustration for solving HJI PDE \eqref{eq:HJPDE}. Here $c_p$ represents the collision point.}
    \label{fig:geometric_argument}
\end{figure}
It was shown in \cite{chacon2019safe} that the collision time computed as above satisfies indeed the HJI PDE \eqref{eq:HJPDE}. The formal result is the following.
\begin{prop}\label{prop:analitical_ttr}
Consider the function $\ttr\left(z\right)$ defined as
\[
\ttr\left(z\right) := \frac{-\left(\pxr\vxr+\pyr\vyr\right) - \sqrt{\Delta}}{\vxr^{2}+\vyr^{2}} \qquad \text{ in } \mathcal{R^{*}}\setminus\domDanger,
\]
where 
\[
\Delta = \left(\pxr\vxr+\pyr\vyr\right)^{2}-\left(\vxr^{2}+\vyr^{2}\right)\left(\pxr^{2}+\pyr^{2}-c_{r}^{2}\right).
\]
Also define $\ttr\left(z\right)$ to be $0$ on $\domDanger$. Then $\ttr\left(z\right)$ satisfies equation \eqref{eq:HJPDE}.
\end{prop}

\begin{proof} We refer to \cite[Prop. 5]{chacon2019safe} for the proof of this result. We only note that the proof there is based on an explicit calculation of the $\arg \min \max$ of the expression in equation \eqref{eq:HJPDE}, which was found to be:
\begin{equation}\label{eq:optimalControlForm1}
    u^{*}=d^{*}=u_{max}\frac{\left(\frac{\partial\ttr\left(z\right)}{\partial \vxr},\frac{\partial\ttr\left(z\right)}{\partial \vyr}\right)}
    {\left\Vert\frac{\partial\ttr\left(z\right)}{\partial \vxr},\frac{\partial\ttr\left(z\right)}{\partial \vyr}\right\Vert}.
\end{equation}
\end{proof}

By implicit differentiation of \eqref{eq:quadraticTTR} we find:
\begin{subequations}
\small
\begin{align}
    \frac{\partial\ttr}{\partial \vxr} &=\frac{-\vxr\ttr^{2}\left(z\right)-\pxr\ttr\left(z\right) }{\left(\vxr^{2}+\vyr^{2}\right)\ttr\left(z\right)+\left(\pxr\vxr+\pyr\vyr\right)}\label{eq:dpsi_dvx} \\
    \frac{\partial\ttr}{\partial \vyr} &=\frac{-\vyr\ttr^{2}\left(z\right)-\pyr\ttr\left(z\right) }{\left(\vxr^{2}+\vyr^{2}\right)\ttr\left(z\right)+\left(\pxr\vxr+\pyr\vyr\right)},\label{eq:dpsi_dvy}
\end{align}
\normalsize
\end{subequations}
and hence, from \eqref{eq:optimalControlForm1} we can derive a closed expression for the optimal avoidance controller. Note that to use this pairwise avoidance strategy we require each vehicle to know its speed and position relative to the other vehicles.

The static HJI PDE \eqref{eq:HJPDE} is typically approximated by finite difference methods such as the one presented in \cite{Yang2013}. Our approach, using an analytic solution, leads to two main advantages.  First, we do not have to deal with large amounts of memory and long computational times involved in refinements of the numerical resolution. Second, while numerical methods can only compute the solution in a bounded domain, an analytical solution allows us to have the best possible resolution in unbounded domains. This allows us to predict and react to possible collisions arbitrarily far into the future.


\subsection{Overall control logic}
\label{subsect:logic}
In this subsection we describe how to switch between the two controllers presented above. 

We will consider that vehicle $\Qi$ is  in potential conflict with vehicle $\Qj$ if the time to collision $\ttr\left(z_{i}\right)$ (here $z_i$ denotes the relative current state of the two vehicles), is less than or equal to a specified time horizon $t_{safety}$. 
In such a case $\Qi$ must use the safety controller, otherwise, the coverage controller is used.

In the case that a vehicle detects more than one conflict, it will apply the control policy of the first conflict detected at that particular time. 
Algorithm \ref{alg:overallControlLogic} describes the overall control logic for a generic vehicle $\Qi$.
\begin{algorithm}
  \caption{Overall control logic for a generic vehicle $\Qi$.}\label{alg:overallControlLogic}
  \textbf{IN}: State $x_i$ of a vehicle $\Qi$; states $\{x_j\}_{j\neq i}$ of other vehicles $\left\{ Q_j\right\}_{j\neq i}$; a domain $\Omega$ to cover.
  \\[2pt] \textbf{PARAMETER:} A time horizon for safety check $t_{safety}$;
  \newline\textbf{OUT}: A control $u_i$ for $\Qi$.
  \label{alg:the_alg}
  \begin{algorithmic}[1]
    \STATE {$safe \leftarrow \text{True}$;}
    \FOR {$j\neq i$} 
        \STATE {$z \leftarrow x_{i}-x_{j}$;}
        \IF {$\ttr\left(z\right) \leq t_{safety}$}
            \STATE {$safe \leftarrow \text{False}$;}
            \STATE {$U_{ix}=-\frac{\vxr\ttr^{2}\left(z\right)+\pxr\ttr\left(z\right)}{\left(\vxr^{2}+\vyr^{2}\right)\ttr\left(z\right)+\left(\pxr\vxr+\pyr\vyr\right)}$;}\label{alg:ln:avoidUx}
            \STATE {$U_{iy}=-\frac{\vyr\ttr^{2}\left(z\right)+\pyr\ttr\left(z\right)}{\left(\vxr^{2}+\vyr^{2}\right)\ttr\left(z\right)+\left(\pxr\vxr+\pyr\vyr\right)}$;}\label{alg:ln:avoidUy}
            \STATE {\textbf{break for;}}
        \ENDIF
    \ENDFOR
    \IF {safe}
        \STATE {$\left(U_{ix},U_{iy}\right)\text{=}  \linebreak\text{-}\sum_{j\neq i}^{N}f_{I}\left(\left\Vert p_{ij}\right\Vert \right)\frac{p_{ij}}{\left\Vert p_{ij}\right\Vert }\text{-}f_{h}\left(\left[\left[h_{i}\right] \right]\right)\frac{h_{i}}{\left[\left[h_{i}\right] \right]}- a_v v_i $;}\label{alg:ln:coverU}
    \ENDIF
    \STATE {$u_{i}=u_{max}\frac{\left(U_{ix},U_{iy}\right)}{\left\Vert\left(U_{ix},U_{iy}\right)\right\Vert}$;}\label{alg:ln:normalizing}
  \end{algorithmic}
  \textbf{RETURN}: $u_{i}$
\end{algorithm}

In Algorithm \ref{alg:overallControlLogic}, lines \ref{alg:ln:avoidUx} and \ref{alg:ln:avoidUy} can be obtained from equations \eqref{eq:optimalControlForm1}, \eqref{eq:dpsi_dvx} and \eqref{eq:dpsi_dvy} (also note the normalization step in line \ref{alg:ln:normalizing}), while line \ref{alg:ln:coverU} comes from the explicit coverage control \eqref{eq:controlExplicit}.

\begin{rmk}\label{rem:threshold-remark}
    By thresholding the force (Algorithm \ref{alg:overallControlLogic} line \ref{alg:ln:normalizing}), the theoretical results may not necessary hold anymore. However, when close to the desired operation point, the coverage forces are small enough to not need to be thresholded, in which case the theoretical results are indeed valid.
\end{rmk}


\subsection{Numerical simulations}
\label{subsect:numerics}

\textbf{Square domain.}
We consider the coverage problem for a square domain. We present two strategies: both use the coverage controller described in Section \ref{subsect:coverage-controller}, but only one strategy switches to the safety controller when necessary, according to Section \ref{subsect:logic}. In both cases 16 vehicles start from a horizontal line setup outside of the target square domain; see the starting locations of the trajectories in Figures \ref{fig:squareNoAvoid5} and \ref{fig:squareAvoid5}.

The simulations from the left and right columns in Figure \ref{fig:square} do not include, and respectively include, the safety controller. The large coloured dots represent the position of the vehicles, the dashed tails are past trajectories (shown for the previous 5 seconds), and the arrows indicate the movement direction. Note that we do not show the arrows when the velocities are too small.

At $t=0\text{ (s)}$ the only contributions come from the vehicle-domain forces, which pull the mobile agents toward the interior of the square; see initial trajectory tails in Figures \ref{fig:squareNoAvoid5} and \ref{fig:squareAvoid5}. At $t=5\text{ (s)}$ the vehicles without safety controller are more prone to collisions due to the symmetry of the initial condition. The safety controller breaks down the symmetry and enables the vehicles to enter the crowded area without collisions.

The presence of overshoots at later times (Figures \ref{fig:squareNoAvoid10} and  \ref{fig:squareAvoid10}) is expected, being due to the piece-wise linear vehicle-domain forces (i.e.,  spring-like forces). However, Figure \ref{fig:squareAvoid10} indicates that in addition to preventing collisions, the use of the safety controller also reduces the overshoots. After $t=50\text{ (s)}$ both control strategies reach a steady state which is an $\rd$-cover for the square . We note that the system with collision avoidance reaches the equilibrium faster; compare Figures \ref{fig:squareNoAvoid50} and \ref{fig:squareAvoid50}.
\begin{figure}
    \begin{subfigure}[b]{0.51\columnwidth}
    \includegraphics[width=\columnwidth]{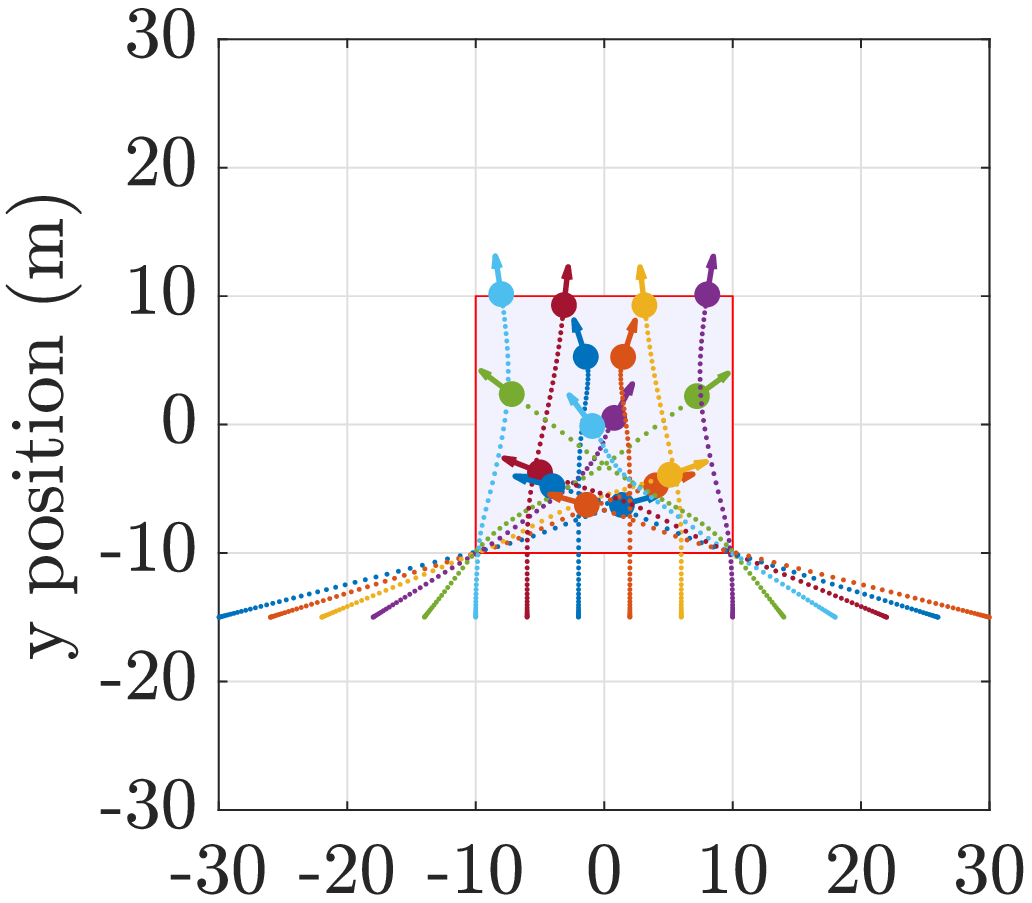}
    \vspace{-2.5em}\caption{$t=5(s)$}
    \label{fig:squareNoAvoid5}
    \end{subfigure}
    \hspace{-1.5em}
    \begin{subfigure}[b]{0.51\columnwidth}
    \includegraphics[width=\columnwidth]{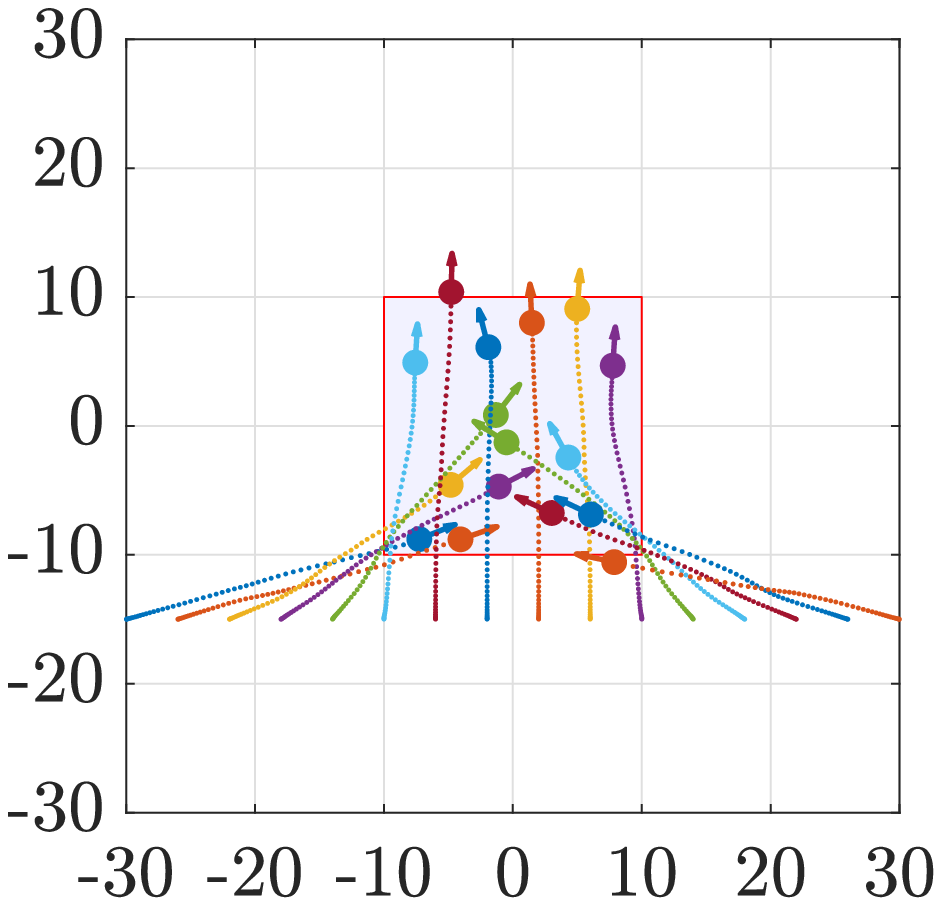}
    \vspace{-2.5em}\caption{$t=5(s)$}
    \label{fig:squareAvoid5}
    \end{subfigure}
    \begin{subfigure}[b]{0.51\columnwidth}
    \includegraphics[width=\columnwidth]{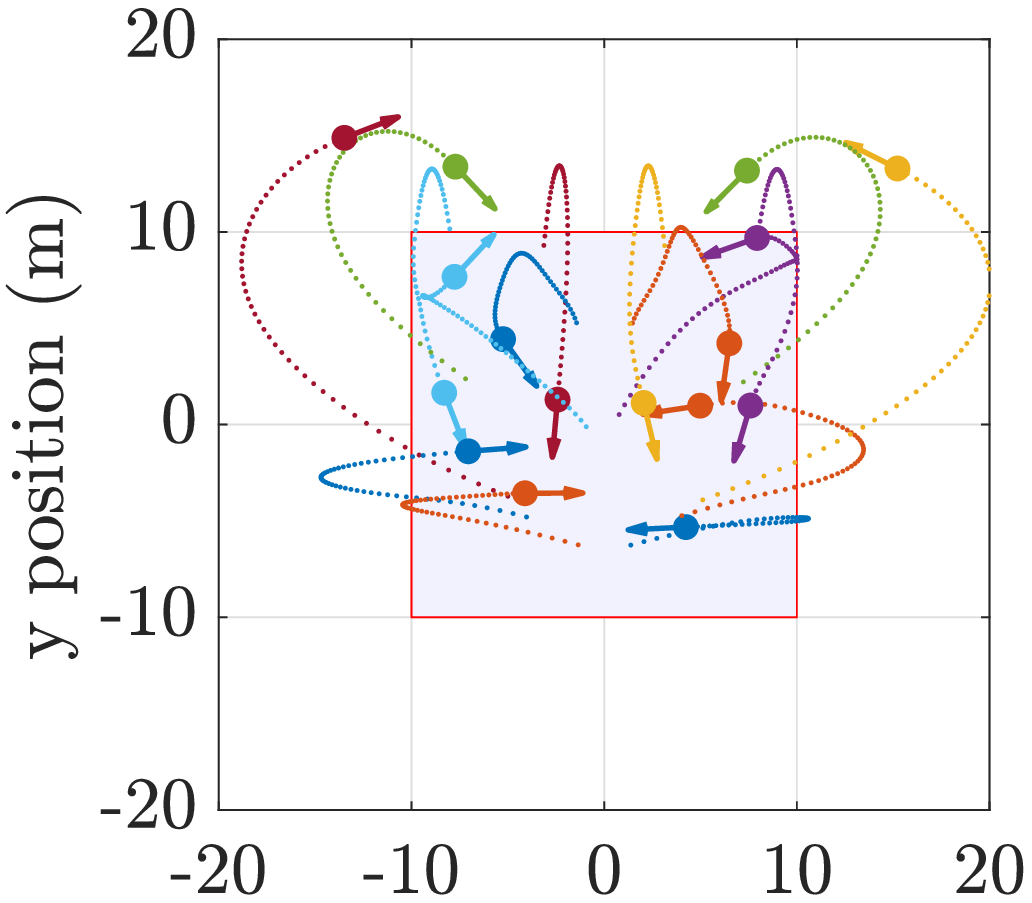}
    \vspace{-2.5em}\caption{$t=10(s)$}
    \label{fig:squareNoAvoid10}
    \end{subfigure}
    \hspace{-1.5em}
    \begin{subfigure}[b]{0.51\columnwidth}
    \includegraphics[width=\columnwidth]{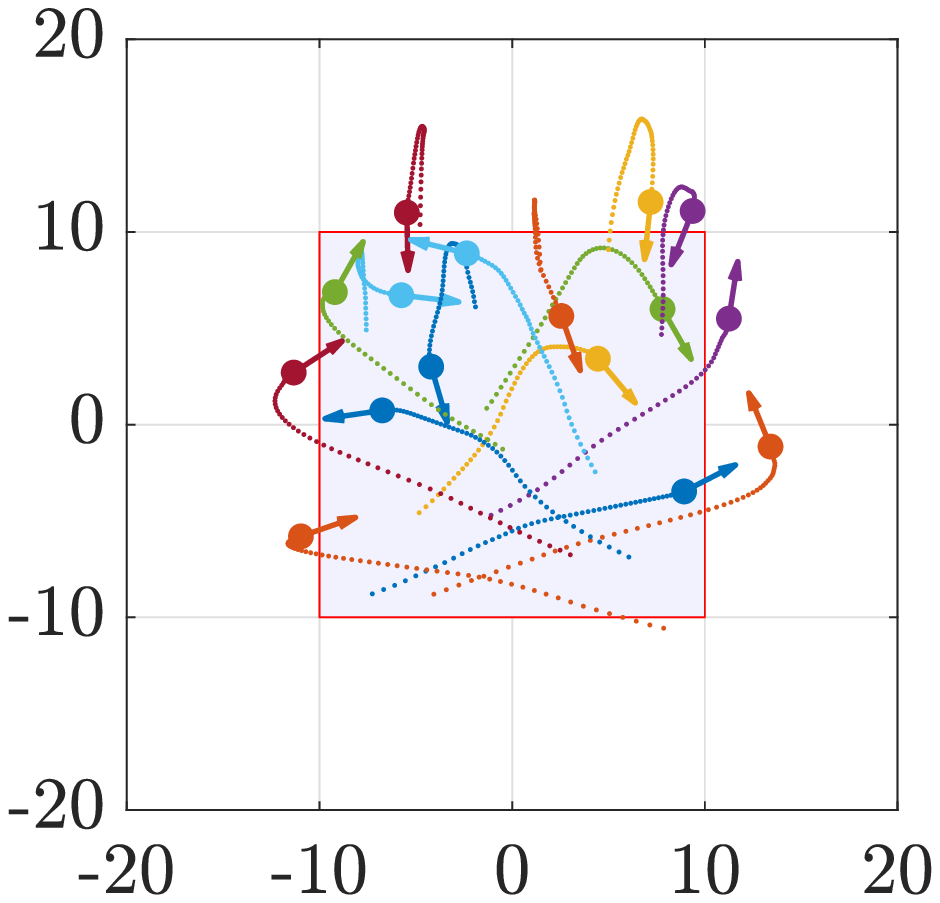}
    \vspace{-2.5em}\caption{$t=10(s)$}
    \label{fig:squareAvoid10}
    \end{subfigure}
    \begin{subfigure}[b]{0.51\columnwidth}
    \includegraphics[width=\columnwidth]{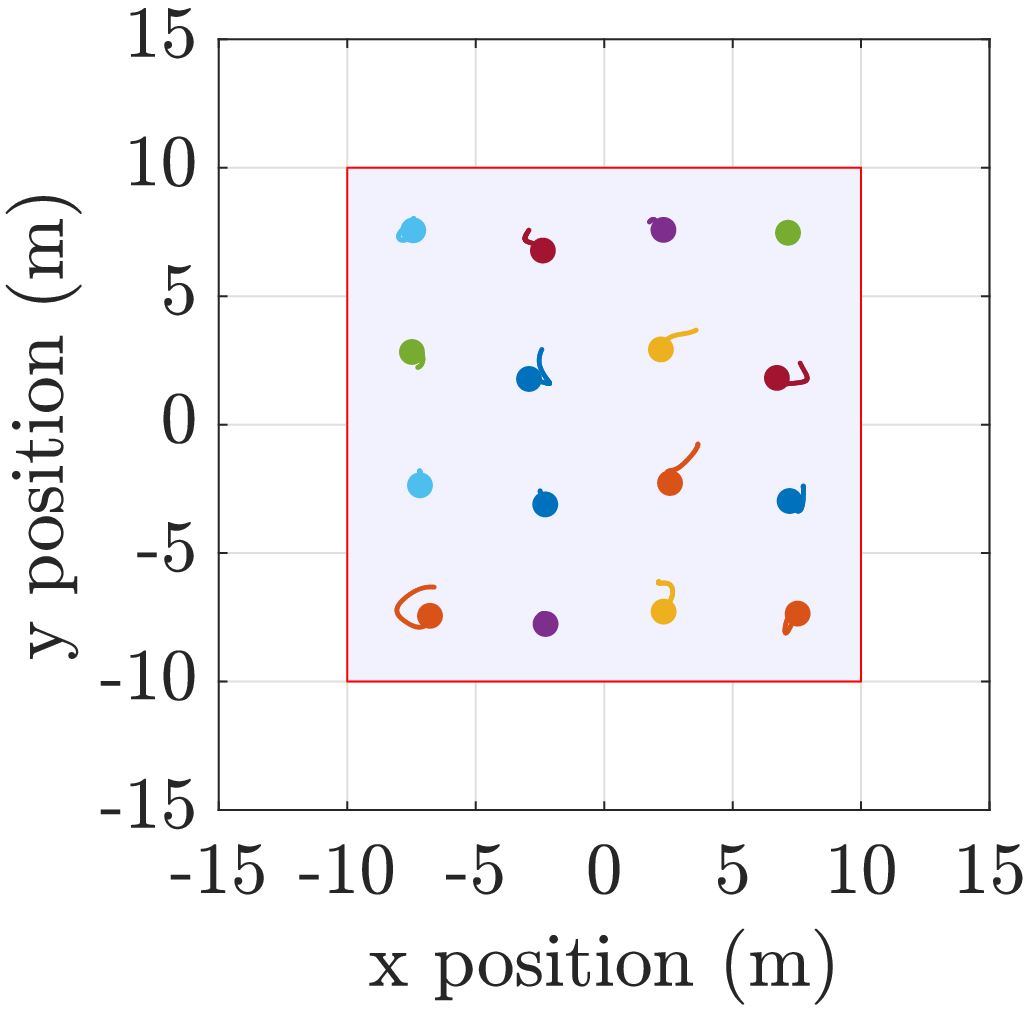}
    \vspace{-1.5em}\caption{$t=50(s)$}
    \label{fig:squareNoAvoid50}
    \end{subfigure}
    \hspace{-1.5em}
    \begin{subfigure}[b]{0.51\columnwidth}
    \includegraphics[width=\columnwidth]{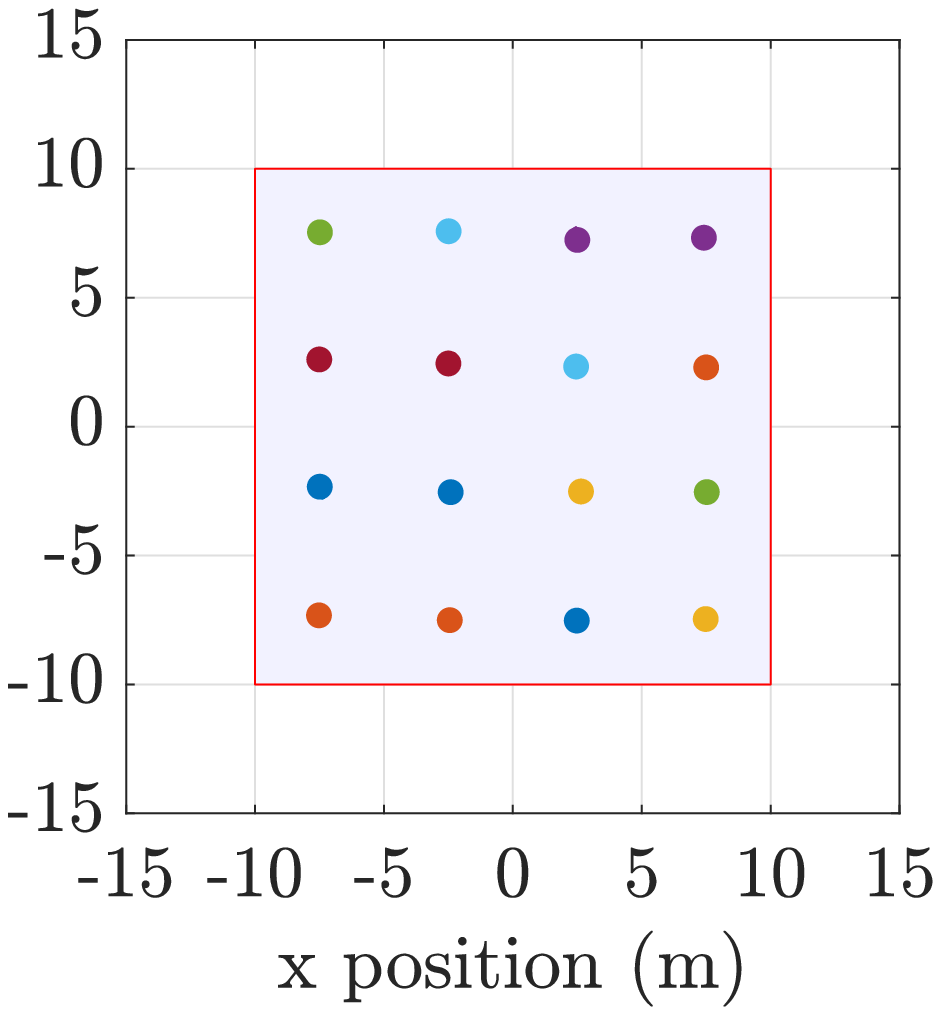}
    \vspace{-1.5em}\caption{$t=50(s)$}
    \label{fig:squareAvoid50}
    \end{subfigure}
    \caption{Square domain coverage at different time instants, without (left) and with (right) safety controller, when $N=16$, $c_{r}=2\,\text{(m)}$, $v_{max}=10\,\text{(m/s)}$, $u_{max}=3 \,\text{(m/s$^{2})$}$, $t_{safety}=5\,\text{(s)}$, side length $l=20\,\text{(m)}$, domain area $A=l^{2}=400\,\text{(m$^2$)}$ and $r_{d}=\sqrt{\frac{A}{N}}=5\,\text{(m)}$. Vehicles start in a horizontal line configuration and reach a square grid steady state which is an $r_d$-cover of the domain (see Definition \ref{defn:cover}). The use of the safety controller reduces both the collision count and the overshoot, and helps reach the steady state faster.} 
    \label{fig:square}
\end{figure}

A \textit{collision event} starts when the distance between two vehicles is less than or equal to the collision radius $c_{r}$, and ends when the distance becomes greater than $c_{r}$. The collision event count for the square domain coverage with and without the safety controller, for various number of vehicles, is shown in Table \ref{tb:collisionsSquare}. We point out that in the absence of the safety controller, the collision count increases significantly with the number of vehicles, while it remains zero or very low when the safety controller is used.

\begin{table}[H]
    \captionsetup{width=\linewidth}
    \caption{Square coverage collision count.}
    \label{tb:collisionsSquare}
    \centering
    \begin{tabular}{ccc}
    number of vehicles & without avoidance & with avoidance\tabularnewline
    \hline
    9 & 8 & 0\tabularnewline
    16 & 51 & 0\tabularnewline
    25 & 146 & 2\tabularnewline
    \hline
    \end{tabular}
\end{table}

\vspace{-1em}

Safety issues may also arise when a vehicle needs to avoid two or more vehicles at the same time. Our safety controller does not guarantee collision avoidance in such cases. Guaranteed collision avoidance for more than two vehicles is an unsolved problem, as explored for example in \cite{Chen2016}.

%% file: moving_domain.tex
\section{Coverage of a moving domain}
\label{sect:moving}


\subsection{Problem formulation}
\label{subsect:problem-align}
We consider now the coverage problem when the target domain moves with prescribed constant velocity $\vd$. Specifically, let $\Omega\subseteq\mathbb{R}^{2}$ be a compact domain and define $\Omegat=\Omega+ t \vd$, representing the moving domain at time $t$. Alternatively, if one sets an arbitrary marker point $\pd$ (e.g., the centre of mass) in $\Omega$, its motion is given by $\pd(t) = \pd + t \vd$. 

We are interested in covering the domain $\Omegat$ (see Definitions \ref{defn:subcover} and \ref{defn:cover}), which changes through time. For this reason we want the vehicles to reach asymptotically, as $t \to \infty$, the velocity of the target domain, while maintaining a cohesive group through dynamics. This is expressed by the concept of flocking \cite{Olfati2006,CuckerSmale2007,HaTadmor2008}.

Consider a group of $N$ vehicles, each of them governed by the double integrator dynamics, i.e.,
\begin{equation}
\label{eq:dynamic}
\dot{\posi}=\veli,\quad \dot{\veli}=\ui, \qquad i =1,\dots,N,
\end{equation}
with control $\ui$ to be specified later. We adapt below the definition of flocking from \cite{HaLiu2009} to the problem of moving target. 

\begin{defn}[Flocking with a moving target]  A group of vehicles has a time-asymptotic flocking with a target domain moving with constant velocity $\vd$ if its positions and velocities $\left\{ \posi,\veli\right\},\,i=1,\cdots,N$ satisfy the following two conditions:
\begin{enumerate}
\item The relative positions with respect to the marker point in the domain are uniformly bounded in time (forming a group):
\[
\underset{0\leq t<\infty}{\sup}\sum_{i=1}^{N}\left\Vert \posi\left(t\right)-\pd\left(t\right)\right\Vert ^{2}<\infty.
\]
\item The relative velocities with respect to the moving domain go to zero asymptotically in time (velocity alignment):
\[
\underset{t\rightarrow+\infty}{\lim}\sum_{i=1}^{N}\left\Vert \veli\left(t\right)-\vd \right\Vert ^{2}=0.
\]
\end{enumerate}
\label{defn:flocking}
\end{defn}

In this case, our safe domain coverage problem of interest is the following:
\medskip

\textbf{Safe-domain-coverage by vehicles with double integrator dynamics for moving domains:} \textit{Consider a compact domain $\Omega_t$ that moves with constant velocity $\vd$ in the plane, and $N$ vehicles with dynamics described by \eqref{eq:dynamic}, starting from safe initial conditions. Find the maximal $r>0$ and a control policy that leads to an $r$-cover for $\Omega_t$ that flocks with the moving target, while satisfying the safety condition \eqref{eq:safety} at any time.}
\medskip


\subsection{Coverage controller with alignment}
\label{subsect:control-align}

Using the same coverage controller \eqref{eq:controlExplicit} for this problem makes the vehicles lag behind the domain, reacting only when they are outside of it. This suggests that the vehicles require a mechanism to align their velocities with that of the target domain, as well as with the velocities of their neighbors.

Inspired by the Cucker-Smale model with rooted leadership (see \cite{ha2014flocking,li2010cucker}), we propose a control force with inclusion of inter-vehicle and vehicle-domain alignment forces, given by:
\begin{equation}
\label{eq:controller-align}
\begin{aligned}
u_{i}&=-\sum_{j\neq i}^{N}f_{I}\left(\left\Vert \pos_{ij}\right\Vert \right)\frac{\pos_{ij}}{\left\Vert \pos_{ij}\right\Vert}
-f_{h}\left(\left[\left[h_{i}\right]\right]\right)\frac{h_{i}}{\left[\left[h_{i}\right]\right]}  \\[-5pt]
& \quad-\underset{\text{inter-vehicle}}{\underbrace{\sum_{j\neq i}^{N}\csf\left(\left\Vert \pos_{ij}\right\Vert \right)\vel_{ij}}} -\underset{\text{vehicle-domain}}{\underbrace{\av\left(\vel_{i}-\vd\right)}}.
\end{aligned}
\end{equation}
Here, 
$\pos_{ij}:=\pos_{i}-\pos_{j}$, $\vel_{ij}:=\vel_{i}-\vel_{j}$,
and  $P_{\partial\Omegat}\left(\pos_{i}\right)$ denotes the projection of $\pos_{i}$ on $\partial\Omegat$. Also, $h_{i}:=\pos_{i}-P_{\partial\Omegat}\left(\pos_{i}\right)$, and $\left[\left[h_{i}\right]\right]$ denotes the signed distance
of $\pos_{i}$ from $\partial\Omegat$.
In addition, $\csf$ is a non-negative communication function and $\av$ is a positive constant.

The inter-vehicle alignment force, which controls the alignment of vehicle $i$'s velocity with the velocities of the rest of the vehicles, depends on the relative distance $\|p_{ij}\|$ between the interacting vehicles. For a communication function $\csf$ that is non-increasing (this is a typical assumption in the literature \cite{HaTadmor2008,HaLiu2009}), vehicles align stronger with their neighbours, and less with vehicles that are further apart. The results presented in this paper correspond to a communication function in the form: 
\[
\csf\left(\left\Vert \pos_{ij}\right\Vert \right)=C_{al} e^{-\frac{\| \pos_{ij}\| }{l_{al}}},
\]
where $C_{al}$ and $l_{al}$ are constants associated to the alignment strength and alignment range, respectively.   This function was considered in \cite{FetecauGuo2012} in the context of honeybee swarms.

The vehicle-domain alignment force drives the velocity of the vehicles to the domain's velocity $\vd$. In this regard, the braking force in the static domain model (see \eqref{eq:controlExplicit}) can also be interpreted as an alignment force that brings the vehicles to a stop. Also, while for simplicity we have taken a common constant $\av$ for all vehicles, the considerations that follow apply to the more general alignment forces 
${a}_{v,i}(\vel_i-v_d)$, with ${a}_{v,i}>0$.

By changing to relative coordinates with respect to the frame of the moving domain, one can recover the case of a stationary domain ($\vd =0$). Indeed, change variables to:
\begin{equation}
\label{eq:ch-var}
\posrel_{i}:=\pos_{i}-t \vd, \quad
\velrel_{i}:=\vel_{i}-\vd,
\end{equation}
and note that the inter-vehicle positions and velocities are invariant to this change of coordinates, i.e.,
\[
\posrel_{ij}  :=\posrel_{i}-\posrel_{j}=\pos_{ij}, \quad 
\velrel_{ij} :=\velrel_{i}-\velrel_{j}=\vel_{ij}.
\]
Also, by translation, the distance to the target domain satisfies 
\begin{align*}
h_{i}&=
\left(\pos_{i}-t \vd \right)-P_{\partial\Omegat-t \vd}\left(\pos_{i}-t \vd\right)\\
&=\posrel_{i}-P_{\partial\Omega}\left(\posrel_{i}\right).
\end{align*}
Hence, in the new variables, the signed distances $[[\hrel_i]]$, where
\begin{equation}
\label{eq:ht}
\hrel_{i} := \posrel_{i}-P_{\partial\Omega}\left(\posrel_{i}\right),
\end{equation}
are with respect to the initial (fixed) domain $\Omega$. 

The observations above allow us to rewrite the control (\ref{eq:controller-align}) in the new variables. We find that in the moving coordinate frame the dynamics of the $N$ vehicles is given by:
\[
\dot{\posrel}_{i} = \velrel_{i}, \quad \dot{\velrel}_{i} = \tilde{u}_i, \qquad i = 1, \dots, N,
\]
where
\begin{align*}
\tilde{u}_{i}&=-\sum_{j\neq i}^{N}f_{I}\left(\left\Vert \posrel_{ij}\right\Vert \right)\frac{\posrel_{ij}}{\left\Vert \posrel_{ij}\right\Vert } -f_{h}([[\hrel_{i}]])\frac{\hrel_{i}}{[[\hrel_{i}]]} \\[-5pt]
&\quad -\sum_{j\neq i}^{N}\csf\left(\left\Vert \posrel_{ij}\right\Vert \right)\velrel_{ij} -\av\velrel_{i}. 
\end{align*}
Note that this corresponds to the dynamics in the original variables for a stationary domain.



\subsection{Asymptotic behaviour}
\label{subsect:Lyapunov}
We first investigate the dynamics with control \eqref{eq:controller-align} for a stationary target ($\vd =0$), using the same interaction functions $f_I$ and $f_h$ from Section \ref{sect:static}, corresponding to potentials \eqref{eq:VI} and \eqref{eq:Vh}.
Consider the same candidate for a Lyapunov function, consisting in kinetic plus (artificial) potential energy:
\[
\Phi=\frac{1}{2}{\sum_{i=1}^{N}}\Bigl(\dot{p}_{i}\cdot\dot{p}_{i}+{\sum_{j\neq i}^{N}}V_{I}\left(p_{ij}\right)+2V_{h}\left(p_{i}\right)\Bigr).
\]
Note that each term in $\Phi$ is non-negative, and $\Phi$ reaches its absolute minimum value when the vehicles are totally stopped. 

The time derivative of $\Phi$ can be calculated as:
\begin{align}
\dot{\Phi} & ={\sum_{i=1}^{N}}\dot{p}_{i}\cdot\Bigl(u_{i}+{\sum_{j\neq i}^{N}}\nabla_{i}V_{I}\left(p_{ij}\right)+\nabla_{i}V_{h}\left(p_{i}\right)\Bigr) \nonumber \\[-7pt]
 & ={\sum_{i=1}^{N}}v_{i}\cdot \biggl(-{\sum_{j\neq i}^{N}} \csf(\|p_{ij}\|)(v_i-v_j) - \av v_i \biggr). \label{eq:phidot-1}
\end{align}

For the inter-vehicle alignment term, write
\small
\begin{multline*}
  {\sum_{i=1}^{N}}v_i \cdot {\sum_{j\neq i}^{N}} \csf(\|p_{ij}\|)(v_i-v_j) = \\ \frac{1}{2}  {\sum_{i=1}^{N}}v_i \cdot {\sum_{j\neq i}^{N}} \csf(\|p_{ij}\|)(v_i-v_j)
  + \frac{1}{2}  {\sum_{j=1}^{N}}v_j \cdot {\sum_{i\neq j}^{N}} \csf(\|p_{ji}\|)(v_j-v_i),
\end{multline*}
\normalsize
where in the second term in the right-hand-side we simply renamed $i \leftrightarrow j$ as indices of summation. From there, use that $\|p_{ij}\| = \|p_{ji}\|$ to get:
\[
  {\sum_{i=1}^{N}} v_i \cdot {\sum_{j\neq i}^{N}} \csf(\|p_{ij}\|)(v_i-v_j) =  \frac{1}{2}  {\sum_{i=1}^{N}} {\sum_{j\neq i}^{N}} \csf(\|p_{ij}\|) \|v_i-v_j\|^2.
\]
Hence, from \eqref{eq:phidot-1}, we find:
\begin{equation*}
\dot{\Phi}  =- \frac{1}{2}  {\sum_{i=1}^{N}} {\sum_{j\neq i}^{N}} \csf(\|p_{ij}\|) \|v_i-v_j\|^2 -\av \sum_{i=1}^{N} \|v_i\|^2.
\end{equation*}

In the case of a target domain moving with velocity $\vd$, one can change to relative coordinates \eqref{eq:ch-var} as explained in Section \ref{subsect:control-align}, and set:
\begin{equation}
\label{eq:Lyapunov-rel}
\Phi=\frac{1}{2}{\sum_{i=1}^{N}}\Bigl(\dot{\posrel}_{i}\cdot\dot{\posrel}_{i}+{\sum_{j\neq i}^{N}}V_{I}\left(\posrel_{ij}\right)+2V_{h}\left(\posrel_{i}\right)\Bigr).
\end{equation}
Then, by the calculations for the stationary target above,
\begin{align}
\dot{\Phi}  &=- \frac{1}{2}  {\sum_{i=1}^{N}} {\sum_{j\neq i}^{N}} \csf(\|\posrel_{ij}\|) \|\velrel_i-\velrel_j\|^2 - \av \sum_{i=1}^{N} \|\velrel_i\|^2 \nonumber \\[-7pt]
&= - \frac{1}{2}  {\sum_{i=1}^{N}} {\sum_{j\neq i}^{N}} \csf(\|p_{ij}\|) \|v_i-v_j\|^2 -\av \sum_{i=1}^{N} \|v_i - \vd\|^2.
 \label{eq:phidot-2}
\end{align}
Note that $\dot{\Phi}$ is negative semidefinite and equal to zero if and only if $\velrel_i=0$ (or equivalently $v_i=\vd$) for all $i$, i.e., when vehicles' velocities are aligned with the velocity of the domain. The construction of this Lyapunov function leads to the following flocking result.

\begin{thm}[Flocking with the moving target]
\label{thm:flocking}
Consider a target domain $\Omega_t$ that moves with constant velocity $\vd$, and a group of $N$ vehicles with smooth dynamics governed by \eqref{eq:dynamic}, with the control law given by \eqref{eq:controller-align}. Then, the group of agents has a time-asymptotic flocking with the moving target $\Omegat$. 
\end{thm}
\begin{proof}
We have to check the conditions in Definition \ref{defn:flocking}. Group cohesiveness (condition 1) can be shown exactly as for Proposition \ref{prop:cohesive}, by using relative coordinates. Indeed, since in relative coordinates the distances to the target are with respect to the fixed domain $\Omega$ (see \eqref{eq:ht}), a similar argument shows that the distances from $\posrel_i(t)$ to the domain $\Omega$ remain bounded by $\sqrt{2 \Phi(0)/\ah}$ when $\posrel_i \notin \Omega$. Restoring the original variables, we can then conclude that there exists $R>0$ such that $\| \posi(t)-\pd(t) \| \leq R $, for all $i$ and $t \geq 0$.

To show velocity alignment (condition 2), we first note that the velocities are also uniformly bounded in time. Indeed, since the potentials $V_I$ and $V_h$ are non-negative and $\Phi$ is non-increasing, we have: 
\[
\sum_{i=1}^{N} \| \velrel_{i}(t) \|^2 \leq 2 \Phi(t) \leq 2 \Phi(0).
\]
Hence, the solutions $\left(\posrel_i(t),\velrel_i(t)\right)$ of the relative system are confined within a compact set through dynamics. By LaSalle Invariance Principle we conclude that the solutions approach asymptotically the largest invariant set in $\{ \dot{\Phi}=0\}$. Consequently, we infer by \eqref{eq:phidot-2} that as $t \to \infty$, the vehicles' velocities approach the velocity of the target domain. 
\end{proof}

\begin{rmk}
\label{rem:equilibria}
The asymptotic states are critical points of $\Phi$ that satisfy $\velrel_{i}=0$ for all $i$. Alternatively, these equilibria are critical points of the artificial potential energy ${\sum_{i=1}^{N}}\Bigl({\sum_{j\neq i}^{N}}V_{I}\left(\posrel_{ij}\right)+2V_{h}\left(\posrel_{i}\right)\Bigr)$. We expect that almost every solution of the relative system will approach asymptotically a local minimum of this potential energy.
\end{rmk}

Most relevant to our study are the $\rd$-covers discussed in Section \ref{sect:static}, now in the context of these configurations being equilibria in the moving frame of the target. Since the potential energy vanishes at such configurations, these relative equilibria are global minimizers. As discussed in the stationary case, in some certain simple geometries, such equilibria are also isolated. For such states, by similar arguments to those used for Proposition \ref{prop:stability}, the following local asymptotic result can be established.
\begin{prop}
\label{prop:stability-al}
Consider a target domain $\Omega_t$ that moves with constant velocity $\vd$, and a group of $N$ vehicles with dynamics defined by \eqref{eq:dynamic} and \eqref{eq:controller-align}. Let the relative equilibrium of interest be of the form $\dot{\posrel}_{i}=0$,  $\left\Vert \posrel_{ij}\right\Vert \geq \rd $  and $[[\hrel_{i}]] \leq -\frac{\rd}{2}$ for  $i,j=1,\cdots,N$ (see Definitions \ref{defn:subcover} and \ref{defn:cover}), and assume that this equilibrium configuration is isolated. Also assume that there is a neighborhood about the equilibrium in which the control law remains smooth. Then, the relative equilibrium is a global minimum of the sum of all the artificial potentials and is locally asymptotically stable.
\end{prop}

\begin{rmk}
\label{rem:no-val} All considerations in this subsection apply to the case of zero inter-individual alignment forces ($\csf=0$). In such case, by working in the moving frame of the domain, the problem reduces in fact to the one studied in Section \ref{sect:static}.
\end{rmk}

As mentioned in Remark \ref{rem:threshold-remark}, the previous theoretical results can only be guaranteed if the control force remains sufficiently small, in other words, if it is threshold free.

\subsection{Numerical simulations}
\label{subsect:num-moving}

In this subsection, we show three numerical simulation scenarios for vehicles using the coverage controller \eqref{eq:controller-align}. While the first two scenarios are covered by the theory, the last one illustrates how our strategy still leads to appropriate final configurations even when the domain follows non-inertial trajectories. The possible safety issues are addressed as described in Section \ref{subsect:collision-controller}.
\medskip

\textbf{Triangular domain.}
We consider the scenario in which an equilateral triangular domain moving with constant velocity, $v_{d}=\left(\frac{\sqrt{2}}{2},\frac{\sqrt{2}}{2}\right)$, is covered by a triangular number of vehicles, i.e. $N=\frac{n\left(n+1\right)}{2},\, n\in\mathbb{N}$. At the start of the simulation the vehicles lie on a line outside the domain (see Figure \ref{fig:moving_triangle_avoid_0}). The evolution for a group of $N=10$ agents, each of them using the coverage with velocity alignment and pairwise safety strategies discussed above, is illustrated in Figures \ref{fig:moving_triangle_avoid_9}-\ref{fig:moving_triangle_avoid_60}. The tails represent the 15-second history of the vehicle positions.

Some of the effects of strong alignments, that is, large $a_v$ or $C_{al}$ values, include vehicles spreading slowly inside the domains or in some cases not reaching the target formation, as pointed out in \cite{chacon2020thesis}.

On the other hand, weak alignments, i.e. small $a_v$ and $C_{al}$ values, cause undesired overshoots, and slower asymptotic flocking. Therefore, it is important to maintain a good balance between the strength of the alignment and coverage forces.

\medskip
\textbf{Non-convex domain.}
We now study the scenario in which vehicles cover and follow a moving non-convex domain in the shape of an arrowhead. While the domain preserves its shape, it moves with a constant velocity $v_{d}=\left(\frac{\sqrt{2}}{2},\frac{\sqrt{2}}{2}\right)$. Different time instants of the simulation are shown in Figure \ref{fig:arrow}, where the tails represent the vehicle positions during the last 20 seconds of the simulation. Initially, all the 9 vehicles lie on a line perpendicular to the movement direction of the target domain, as shown by the tails of the vehicles in Figure \ref{fig:arrow18}. 

We distinguish two main behaviours: during a first phase of the simulation (Figure \ref{fig:arrow18}) the vehicles cover the domain approximately evenly, adopting the arrow shape, while in a second phase (Figure \ref{fig:arrow60}), a clearer domain-following behaviour is observed. The oscillations of the two vehicles that are lagging behind are the effect of their proximity to the corners. Indeed, as one of the line segments of the boundary wedge gets closer to the vehicle near the corner, it pushes it towards the other segment of the wedge, a back-and-forth motion that causes the zigzagging. These oscillations can be reduced by reinforcing the velocity alignment.

Unlike the convex case, in non-convex domains the projection on the boundary for points outside of the domain may not be unique; this is the case for instance of the green vehicle in the middle of the initial setup -- see start of the tails in Figure \ref{fig:arrow18}. Although the chance for a vehicle to lie in one of these states is extremely unlikely (the set of points where this happens has zero measure), this fact may yield ambiguity in the definition of the domain-vehicle force. 
We mitigate this issue by considering the contribution from only one of the multiple projection points; consequently, the numerical time evolution may depend on the chosen projection method.

\begin{figure}
\centering
\begin{subfigure}[b]{0.52\columnwidth}
\includegraphics[width=\columnwidth]{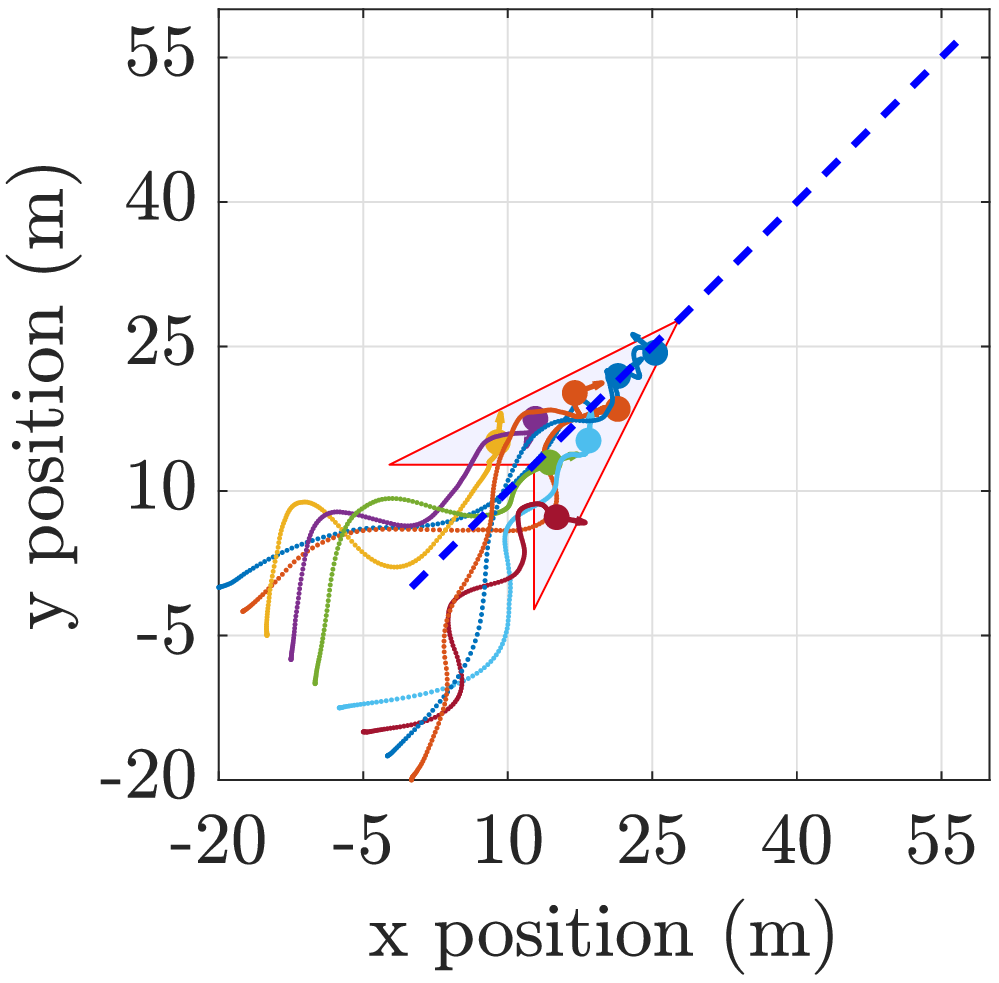}
\caption{$t=18 (s)$}
\label{fig:arrow18}
\end{subfigure}
\hspace{-1.7em}
\begin{subfigure}[b]{0.52\columnwidth}
\includegraphics[width=\columnwidth]{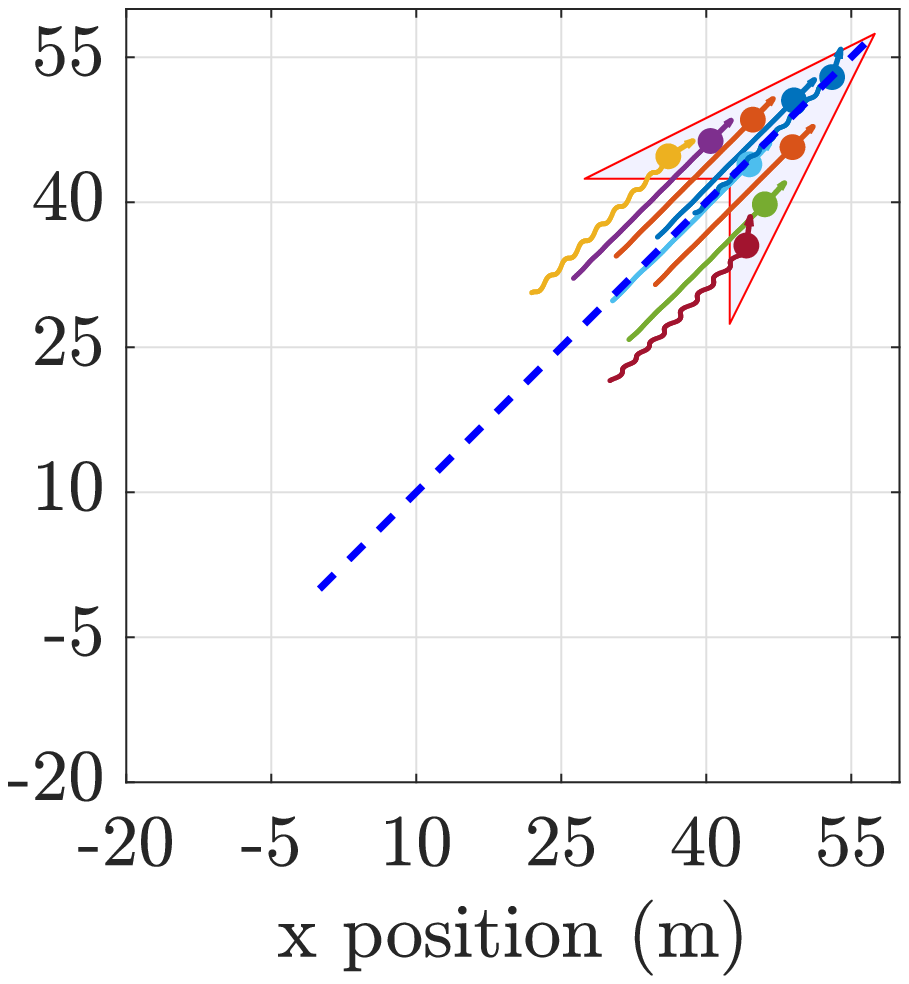}
\caption{$t=60 (s)$}
\label{fig:arrow60}
\end{subfigure}
\caption{Vehicles covering and following a moving, non-convex domain, when $N=9$, $c_{r}=2\,\text{(m)}$, $v_{max}=10\,\text{(m/s)}$, $u_{max}=3 \,\text{(m/s$^{2})$}$, $t_{safety}=5\,\text{(s)}$, $\aI=1\,\text{(m/s$^{2})$}$, $\ah=2\,\text{(m/s$^{2})$}$, $\av=0.2\,\text{(m/s$^{2})$}$, $C_{al}=0.1\,\text{(m/s$^{2})$}$, $l_{al}=7.21\,\text{(m)}$, domain area $A=225\,\text{(m$^2$)}$ and $r_{d}=\sqrt{\frac{A}{N}}=5\,\text{(m)}$. The vehicles start in linear formation, approach and cover the domain, while following it. The vehicles lagging behind exhibit oscillations due to a bouncing effect in the narrow corners.}
\label{fig:arrow}
\end{figure}

\medskip
\textbf{Domain moving in a circle.}
Finally, we include the case of a target domain moving with non-zero acceleration, more specifically, an equilateral triangular domain moving on a circular path. The triangular domain moves so that its centre of mass describes a circular motion of radius 30 with constant angular velocity $\frac{2\pi}{40}$, while aligning its heading to be tangent to this circle (see Figure \ref{fig:circle}). Note that this non-inertial path is not covered by our previous theoretical results (Theorem \ref{thm:flocking} and Proposition \ref{prop:stability-al}).

The vehicles' time evolution is illustrated in Figure \ref{fig:circle}, where the tails represent the vehicle positions during the last 20 seconds of the simulation. At the beginning, the $N=6$ vehicles are in the line formation as shown by the beginning of the tails in Figure \ref{fig:circle9}. As in previous simulations, the vehicles try to reach the moving domain, this time rotating around the domain's circular path (Figures \ref{fig:circle9} and \ref{fig:circle24}). Then, the vehicles reach coverage of the domain (Figure \ref{fig:circle48}) which is maintained by each vehicle by remaining in a circular movement of constant radius (Figure \ref{fig:circle70}).

When a vehicle describes a uniform circular movement with  angular velocity $\omega$ and radius $r$, its speed remains constant over time and is given by $r\omega$. As the vehicles move asymptotically along circles with different radii, they have different velocities, and hence, this type of "flock" does not satisfy Definition \ref{defn:flocking}. In contrast to the case when the domain is moving along inertial paths, in this case each vehicle's control force magnitude does not go asymptotically to zero, but it approaches its centripetal acceleration $r\omega^{2}$ instead.

\begin{figure}
\centering
\begin{subfigure}[b]{0.52\columnwidth}
\includegraphics[width=\columnwidth]{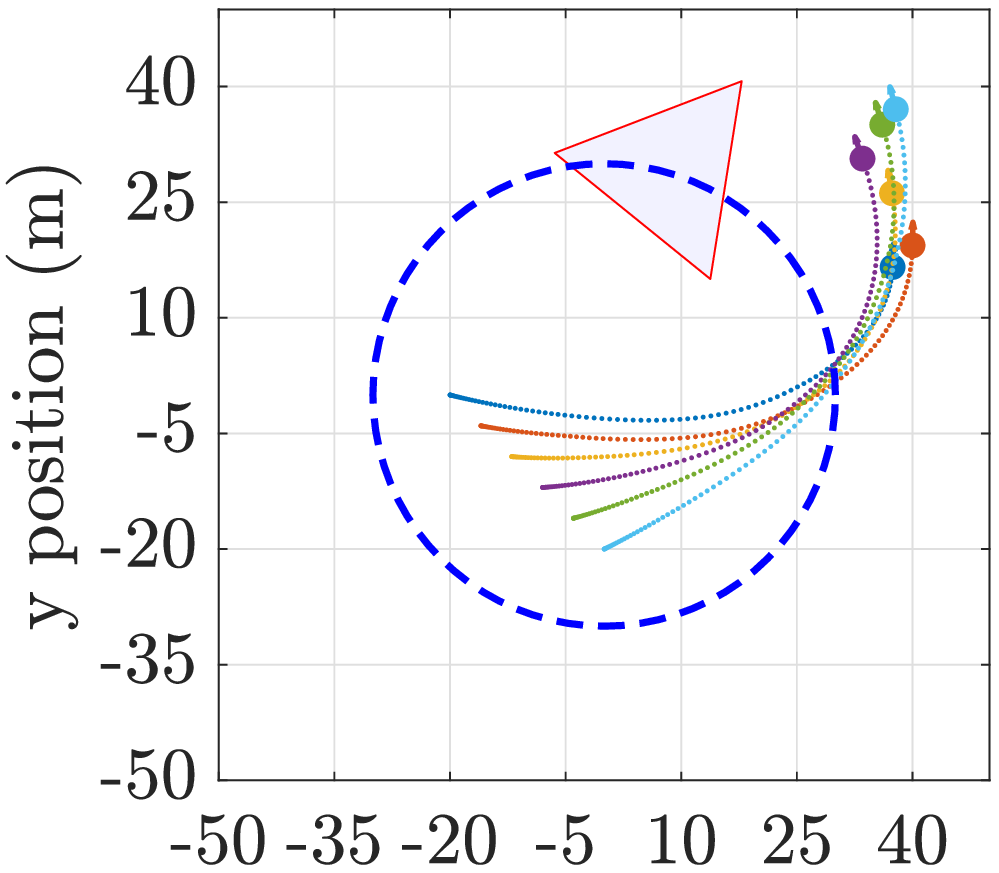}
\vspace{-2.5em}\caption{$t=9(s)$}
\label{fig:circle9}
\end{subfigure}
\hspace{-1.7em}
\begin{subfigure}[b]{0.52\columnwidth}
\includegraphics[width=\columnwidth]{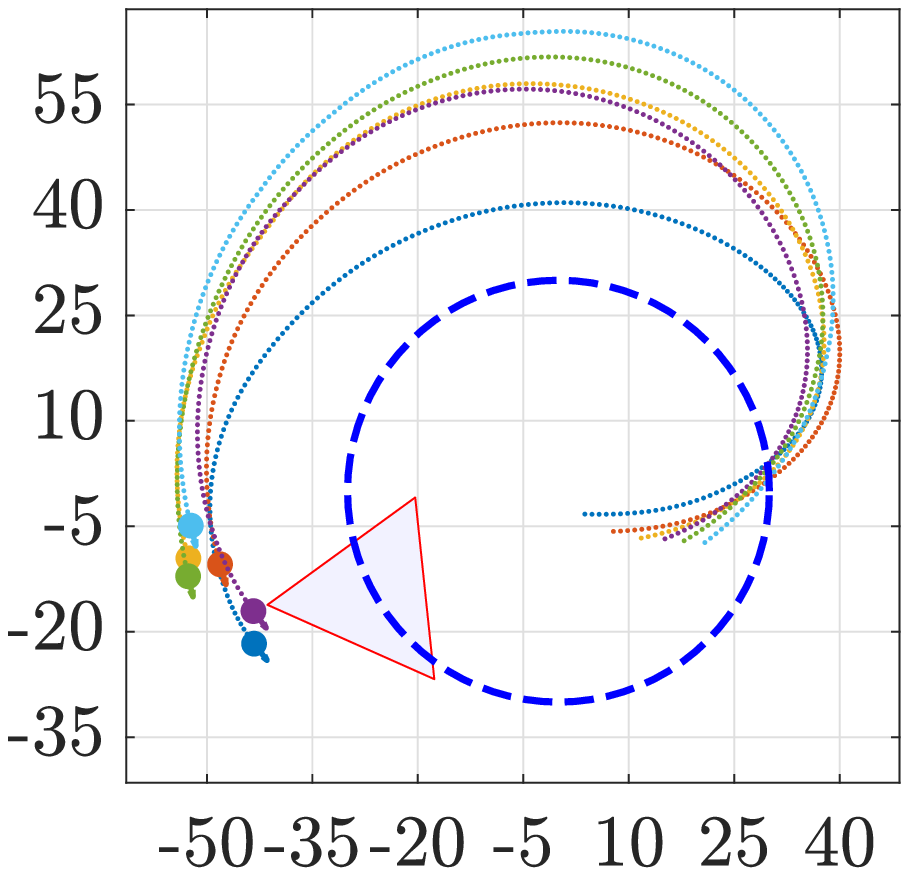}
\vspace{-2.5em}\caption{$t=24(s)$}
\label{fig:circle24}
\end{subfigure}
\begin{subfigure}[b]{0.52\columnwidth}
\includegraphics[width=\columnwidth]{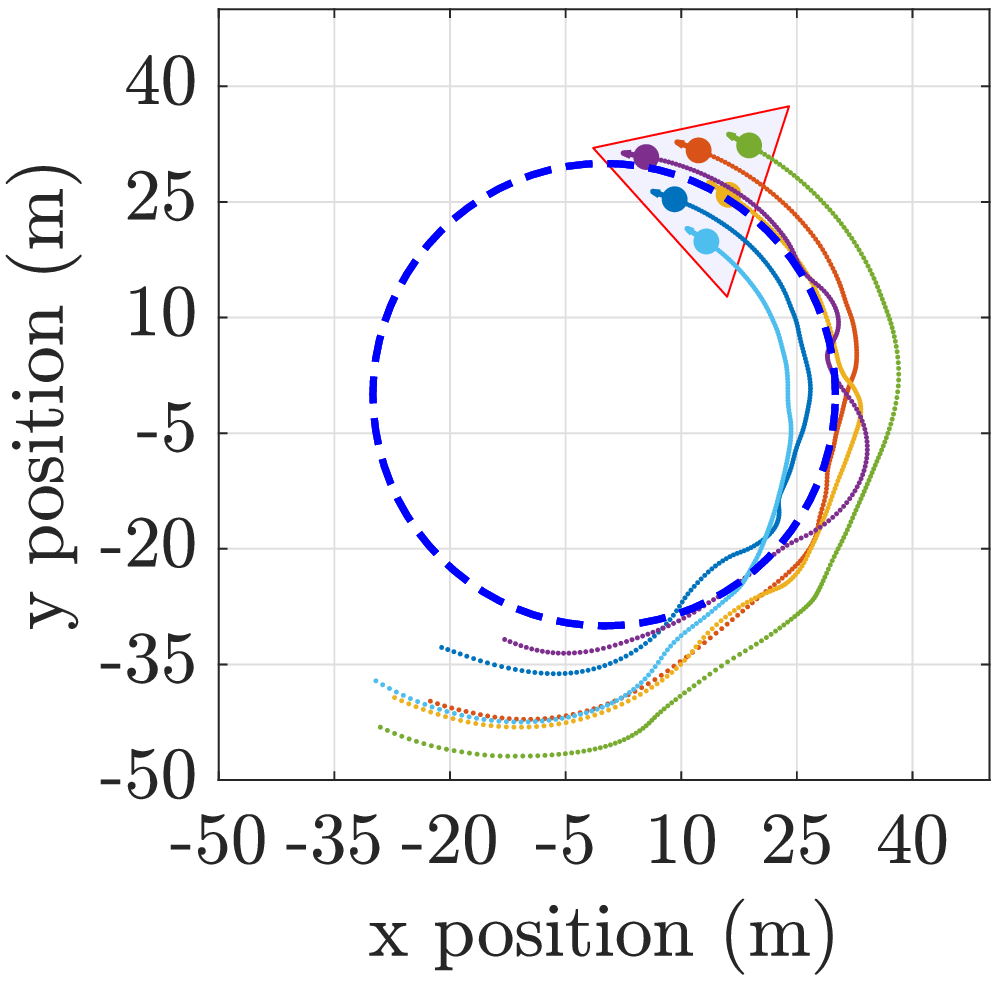}
\vspace{-1.5em}\caption{$t=48(s)$}
\label{fig:circle48}
\end{subfigure}
\hspace{-1.7em}
\begin{subfigure}[b]{0.52\columnwidth}
\includegraphics[width=\columnwidth]{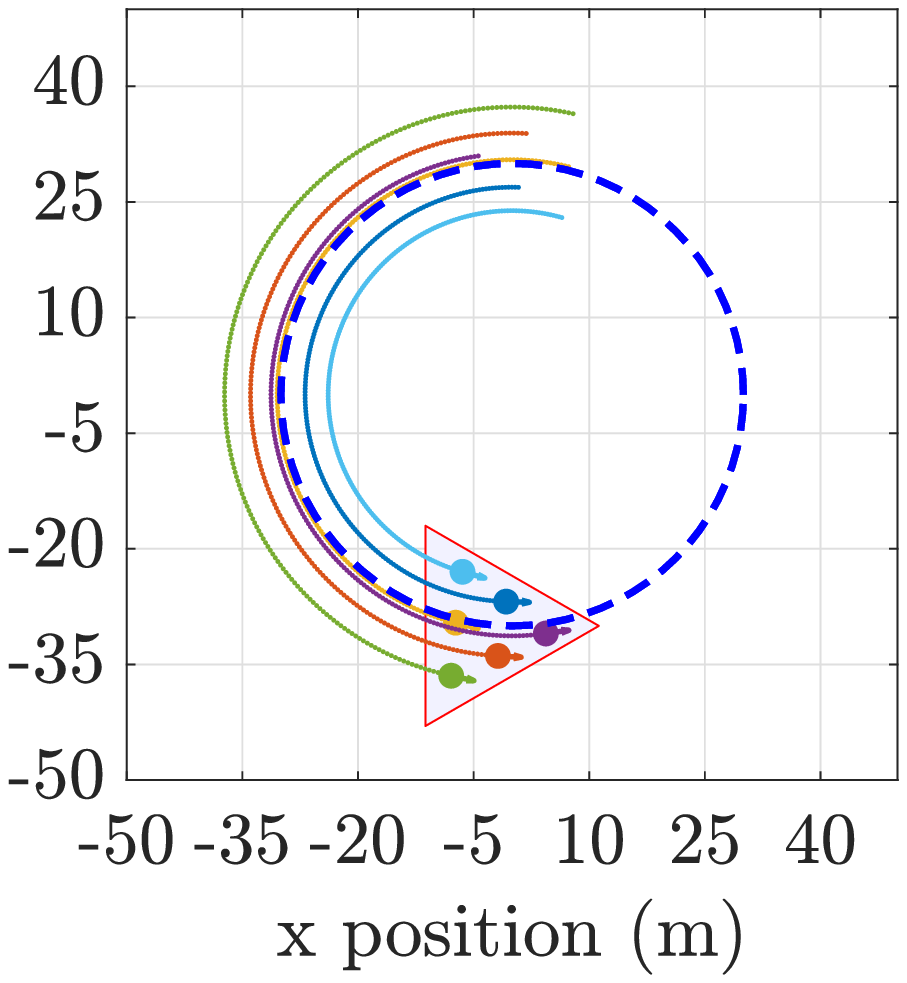}
\vspace{-1.5em}\caption{$t=70(s)$}
\label{fig:circle70}
\end{subfigure}

\caption{Vehicles covering and following a non-zero acceleration triangular domain moving over the path $\left(30\cos\left(\frac{2\pi}{40} t\right),30\sin\left(\frac{2\pi}{40} t\right)\right)$. Here, $N=6$, $c_{r}=2\,\text{(m)}$, $v_{max}=10\,\text{(m/s)}$, $u_{max}=3 \,\text{(m/s$^{2})$}$, $t_{safety}=5\,\text{(s)}$, $\aI=10\,\text{(m/s$^{2})$}$, $\ah=10\,\text{(m/s$^{2})$}$, $\av=1\,\text{(m/s$^{2})$}$, $C_{al}=1.2\,\text{(m/s$^{2})$}$, $l_{al}=10.06\,\text{(m)}$, domain area $A=292.28\,\text{(m$^2$)}$ and $r_{d}=\sqrt{\frac{A}{N}}=6.97\,\text{(m)}$. The vehicles start in linear formation, approach and cover the domain, while following it.}
\label{fig:circle}
\end{figure}

%% file: fixed_wing.tex
\section{Planar Fixed-Wing Aircraft}
\label{sect:fixed-wing}
\subsection{Problem formulation}

In this section, we consider the flocking coverage problem for $N$ vehicles governed by the planar fixed-wing aircraft dynamics, given by:
\begin{equation}\label{eq:dubins_car_dyn}
\begin{aligned}
       \dot{\pos_{i}}=s_{i}\left(\cos\left(\heading_{i}\right),\sin\left(\heading_{i}\right)\right),\,\left(\dot{\heading}_{i},\dot{\speed}_{i}\right)=\left(u_{i,\heading},u_{i,\speed}\right);\\
       0<\speed_{min}\leq\left\Vert \dot{\pos_{i}}\right\Vert \leq\speed_{max},\,\left|u_{i,\heading}\right|\leq u_{\heading_{max}},\,\left|u_{i,\speed}\right|\leq u_{\speed_{max}}.
\end{aligned}    
\end{equation}

\noindent Here, $\heading_{i}$ is the heading angle, $\speed_{i}$
is the vehicle speed and $\pos_{i}=\left(\pos_{i,x},\pos_{i,y}\right)$ is the position of the $i$-th agent. The variables $u_{i,\speed}$ and $u_{i,\heading}$ are the  acceleration and turn rate applied to this vehicle respectively; these are the control inputs to be specified later.
In addition to the bounds for the controls, we also impose maximum and minimum speed limits,
the latter being particularly relevant for aerial vehicles.

In this case, our safe domain coverage problem of interest is the same as for the double integrator dynamics, but considering fixed-wing vehicles.

\subsection{Coverage controller and safety controller}\label{sec:fiked_wing_coverage_control}

The goal is to find expressions for each agent's control law based on the proposed coverage policy with inter-vehicle and vehicle-domain alignment forces \eqref{eq:controller-align}, while satisfying the constraints given in \eqref{eq:dubins_car_dyn}.

By differentiation of the vehicle dynamics \eqref{eq:dubins_car_dyn} with  respect to time, one can find that the acceleration in Cartesian coordinates of the $i$-th agent in terms of the control inputs are as follows:
\begin{equation}
\label{eq:car_to_quad_force_transformation}
\left(\begin{array}{c}
\ddot{\pos_{i,x}}\\
\ddot{\pos_{i,y}}
\end{array}\right)=R\left(\heading_{i},\speed_{i}\right)\left(\begin{array}{c}
u_{i,\heading}\\
u_{i,\speed}
\end{array}\right),
\end{equation} 
where $R\left(\heading_{i},\speed_{i}\right):=\left(\begin{array}{cc} -\speed_{i}\sin\left(\heading_{i}\right) & \cos\left(\heading_{i}\right)\\ \speed_{i}\cos\left(\heading_{i}\right) & \sin\left(\heading_{i}\right)\end{array}\right).$
\smallskip

This relation allows us to compute an expression for the vehicle control inputs in terms of its acceleration in Cartesian coordinates whenever
$s_{i}\neq0$:
\begin{equation}\label{eq:inverse_transformation}
\left(\begin{array}{c}
u_{i,\heading}\\
u_{i,\speed}
\end{array}\right)=\left(R\left(\heading_{i},\speed_{i}\right)\right)^{-1}\left(\begin{array}{c}
\ddot{\pos_{i,x}}\\
\ddot{\pos_{i,y}}
\end{array}\right).
\end{equation}
Using this correspondence, one can obtain the necessary controls $\left(u_{i,\heading},u_{i,\speed}\right)$ to achieve the same acceleration in Cartesian coordinates produced by the proposed control force \eqref{eq:controller-align} as:
\small
\begin{multline}\label{eq:dubins_full_control}
\left(\begin{array}{c}
u_{i,\heading}\\
u_{i,\speed}
\end{array}\right)=\left(R\left(\heading_{i},\speed_{i}\right)\right)^{-1}\left(-\sum_{j\neq i}^{N}f_{I}\left(\left\Vert \pos_{ij}\right\Vert \right)\frac{\pos_{ij}}{\left\Vert \pos_{ij}\right\Vert }\right.\\
\left.-f_{h}\left(\left[\left[h_{i}\right]\right]\right)\frac{h_{i}}{\left[\left[h_{i}\right]\right]}-\sum_{j\neq i}^{N}\csf\left(\left\Vert \pos_{ij}\right\Vert \right)\vel_{ij}-a\left(\vel_{i}-\vd\right)\right).
\end{multline}
\normalsize
The changes of coordinates \eqref{eq:car_to_quad_force_transformation} and \eqref{eq:inverse_transformation} guarantee that all the stability results in Chapter \ref{sect:moving} are still valid in the case of the planar fixed-wing aircraft model when no constraints are applied, as long as none of the vehicles stop along their trajectories. This seems to be a very plausible assumption in practice, as the minimum speed is supposed to be greater than zero.
\medskip

\textbf{Thresholding the control force.}
While finding an admissible force satisfying the constraints for the double integrator model \eqref{eq:dynamic} is done by simply normalizing the vehicles' control input \eqref{eq:controller-align}, obtaining a suitable fixed-wing control input satisfying the constraints \eqref{eq:dubins_car_dyn} is not as straightforward.

In order to obtain the appropriate fixed-wing aircraft control inputs we use
relation \eqref{eq:car_to_quad_force_transformation}, which allows us to represent the set of admissible accelerations from the Cartesian perspective:

\footnotesize
\[
S\left(\heading_{i},\speed_{i}\right)=\left\{ R\left(\heading_{i},\speed_{i}\right)\left(\begin{array}{c}
u_{\heading}\\
u_{\speed}
\end{array}\right):\left(\begin{array}{c}
u_{\heading}\\
u_{\speed}
\end{array}\right)\in\begin{array}{c}\left[-u_{\heading_{max}},u_{\heading_{ max}}\right]\\\times\left[-u_{\speed_{max}},u_{\speed_{ max}}\right]\end{array}\right\}. 
\]
\normalsize
This set can be understood as a stretch and rotation of the rectangle containing the admissible vehicle control inputs.

The input constraints affect the magnitude of the vehicle's acceleration, however, we intend to preserve its direction. Let us define the Cartesian admissible force associated to the control \eqref{eq:dubins_full_control} as 
\begin{equation*}
\left(\begin{array}{c}
\hat{u}_{i,x}\\
\hat{u}_{i,y}
\end{array}\right)= \tau\left(\heading_{i},\speed_{i}\right)R\left(\heading_{i},\speed_{i}\right)\left(\begin{array}{c}
u_{i,\heading}\\
u_{i,\speed}
\end{array}\right),
\end{equation*}
where
\small
\begin{equation*}
\tau\left(\heading_{i},\speed_{i}\right)=\sup\left\{ t\in\mathbb{R}:tR\left(\heading_{i},\speed_{i}\right)\left(\begin{array}{c}
u_{i,\heading}\\
u_{i,\speed}
\end{array}\right)\in S\left(\heading_{i},\speed_{i}\right)\right\}.
\end{equation*}
\normalsize
In other words, $\hat{u}_i$ us the largest acceleration in the set of admissible accelerations in Cartesian coordinates that is parallel to the desired acceleration -- see Figure \ref{fig:Thresholding-Dubins-car}.  

Finally, the thresholded fixed-wing aircraft control inputs can be obtained by using \eqref{eq:car_to_quad_force_transformation} as
\[
\left(\begin{array}{c}
\hat{u}_{i,\heading}\\
\hat{u}_{i,\speed}
\end{array}\right)=\left(R\left(\heading_{i},\speed_{i}\right)\right)^{-1}\left(\begin{array}{c}
\hat{u}_{i,x}\\
\hat{u}_{i,y}
\end{array}\right)=\tau\left(\heading_{i},\speed_{i}\right) \left(\begin{array}{c}
u_{i,\heading}\\
u_{i,\speed}
\end{array}\right).
\]

\begin{figure}
\begin{centering}
\includegraphics[width=0.7\columnwidth]{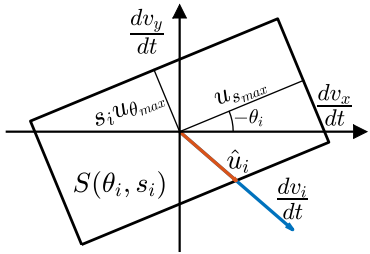}
\par\end{centering}
\caption{\label{fig:Thresholding-Dubins-car} Thresholded fixed-wing aircraft control input $\hat{u}_i$ for vehicle with speed $\speed_{i}$ and heading $\heading_{i}$ computed from its set of admissible accelerations in Cartesian coordinates $S\left(\heading_{i},\speed_{i}\right)$ and a reference acceleration $\frac{dv_{i}}{dt}$.}

\end{figure}



\medskip
\textbf{Collision avoidance.}
As mentioned above, the changes of coordinates \eqref{eq:car_to_quad_force_transformation} and \eqref{eq:inverse_transformation} guarantee that the unconstrained fixed-wing aircraft dynamics satisfies the vehicles safety conditions when the initial energy is small enough, as established in Proposition \ref{prop:collision-avoid}.


As the set of control inputs and non-anticipative disturbances differ from those in the double integrator dynamics, the collision avoidance via Hamilton-Jacobi reachability for this type of vehicles requires a different study  than the one carried out in Section \ref{subsect:collision-controller}.

\begin{figure}
    \centering
    \includegraphics[width=0.8\columnwidth]{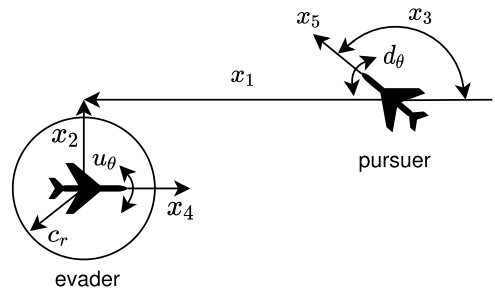}
    \caption{Relative coordinate system for pairwise collision avoidance in the fixed-wing model. Figure adapted from \cite{mitchell2005time} \textcopyright\hspace{0.3em}[2005] IEEE with authors' permission.}
    \label{fig:relative-coordinate-system-fixed-wing}
\end{figure}

Similarly to the double integrator case, we consider the relative dynamics between a pair of vehicles, where one acts as evader and the other as pursuer. Using as reference Figure \ref{fig:relative-coordinate-system-fixed-wing}, to obtain the relative dynamics we consider the evader fixed at the origin, and facing the positive $x_1$ axis. In the figure, $x_1$ is the projection of the vector connecting the vehicles' positions on the axis parallel to the evader's heading, and $x_{2}$ is its projection on the orthogonal direction. Also, $x_{3}$ represents the difference of the two agents' headings, while $x_{4}$ and $x_{5}$ represent the evader and pursuer speeds, respectively. 

The pursuer's speed, relative location and heading, along with the evader's speed, are described by the following dynamical system:
\[
\dot{x}=f\left(x,u,d\right)=\left(\begin{array}{c}
x_{5}\cos\left(x_{3}\right)-x_{4}+u_{\theta}x_{2}\\
x_{5}\sin\left(x_{3}\right)-u_{\theta}x_{1}\\
d_{\theta}-u_{\theta}\\
u_{s}\\
d_{s}
\end{array}\right).
\]
Here $u:=\left(u_{\heading},u_{\speed}\right)$ and $d:=\left(d_{\heading},d_{\speed}\right)$, where  $u_{\heading}$, $u_{\speed}$  and $d_{\heading}$, $d_{\speed}$ are the evader's, respectively the pursuer's, turn rate and acceleration, with the latter being treated as disturbances.

By the dynamic programming principle, the time $\phi$ to reach collision is the viscosity solution for the stationary HJ PDE \eqref{eq:HJPDE} where
\[\domDanger=\left\{ z:z_{1}^{2}+z_{2}^2\leq c_{r}^{2}\text{ or }z_{4} < \speed_{min} \text{ or }z_{4} > \speed_{max} \right\}\]  
is the union of a five dimensional cylinder of radius $c_r$ on the first two dimensions, with two half-spaces, and 
\[\domSafe=\left\{ z:z_{5}<\speed_{min} \text{ or } z_{5}>\speed_{max}\right\}\]
is the union of two half-spaces.

Note that unlike the typical time-to-reach setting, we consider $\domDanger$ to be the set of dangerous states, and $\domSafe$ to be the set of unsafe states that ``invalidates" danger.
Specifically, $\domSafe$ ensures that the pursuer does not violate its speed constraints while trying to cause a collision.

Obtaining an analytical solution for the HJ PDE associated to this problem seems to be much more difficult than in the double integrator case, and we opt for solving it numerically.

The challenges of solving numerically this HJ PDE are two-fold. First, the memory requirements to store the solution are large even for coarse resolutions, and second, the computational time scales poorly as the grid size grows. This is particularly problematic when the algorithm specifications are not optimized for the hardware architecture. We alleviate the second issue by using the new python toolbox \url{https://github.com/SFU-MARS/optimized_dp} for solving HJ PDEs, which yields faster executions by decoupling the algorithm from the hardware specifications. 


\subsection{Numerical simulations} \label{subsect:num-fixed-wing}
In this subsection we consider two simulation scenarios related to those investigated in Section \ref{subsect:num-moving}, and illustrate how our control strategy leads to similar coverage configurations. We do not intend to compare the systems' evolution, as they are two distinct types of vehicles with different capabilities.
\medskip

\textbf{Triangular domain.}
In the first scenario, we consider an equilateral triangular domain moving with constant velocity $v_{d}=\left(\frac{\sqrt{2}}{2},\frac{\sqrt{2}}{2}\right)$, which is covered by a triangular number of vehicles. The minimum for vehicles' speed is set at $0.5\,\text{(m/s)}$, while the maximum is $5\,\text{(m/s)}$. Each of the fixed-wing agents uses the coverage controller with velocity alignment \eqref{eq:dubins_full_control} discussed in Subsection \ref{sec:fiked_wing_coverage_control}.

Figure \ref{fig:db_moving_triangle} shows four different time steps of the evolution of the vehicles. The tails represent the last 15 seconds of the vehicles' position history. At the start of the simulation the $N=10$ vehicles lie on a line outside the domain, moving with the minimum allowed speed and random headings (see Figure \ref{fig:db_moving_triangle_0}). As time evolves, the vehicles approach the domain (Figure \ref{fig:db_moving_triangle_2_4}), and then cover it by taking a triangular formation moving with constant velocity as expected, see Figures \ref{fig:db_moving_triangle_12} and \ref{fig:db_moving_triangle_48}. 

In this particular case, the same coverage controller parameters used for the double integrator vehicles seem to work well. However, it is not a rule of thumb, as the thresholding strategies are very different. We also note that the collisions count goes from 2, when no collision avoidance is included, to 0, when the safety controller is used.

\begin{figure}
\centering
\begin{subfigure}[b]{0.52\columnwidth}
\includegraphics[width=\columnwidth]{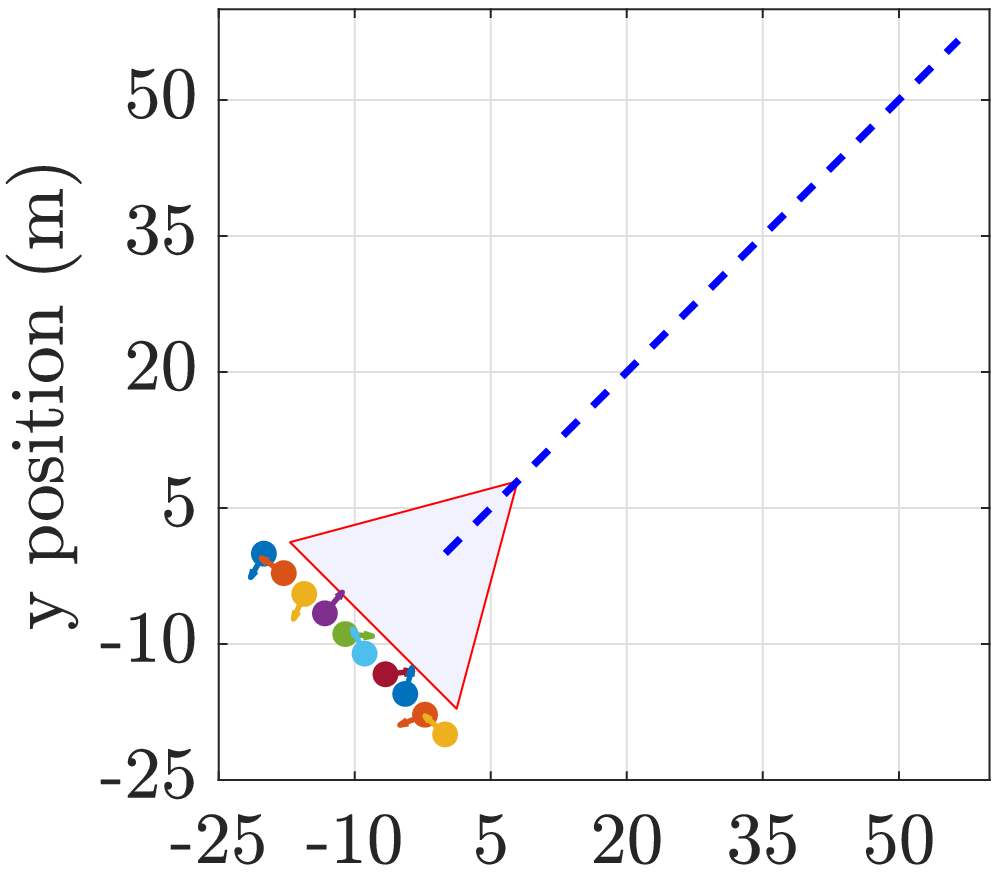}
\vspace{-2.5em}\caption{$t=0(s)$}
\label{fig:db_moving_triangle_0}
\end{subfigure}
\hspace{-1.7em}
\begin{subfigure}[b]{0.52\columnwidth}
\includegraphics[width=\columnwidth]{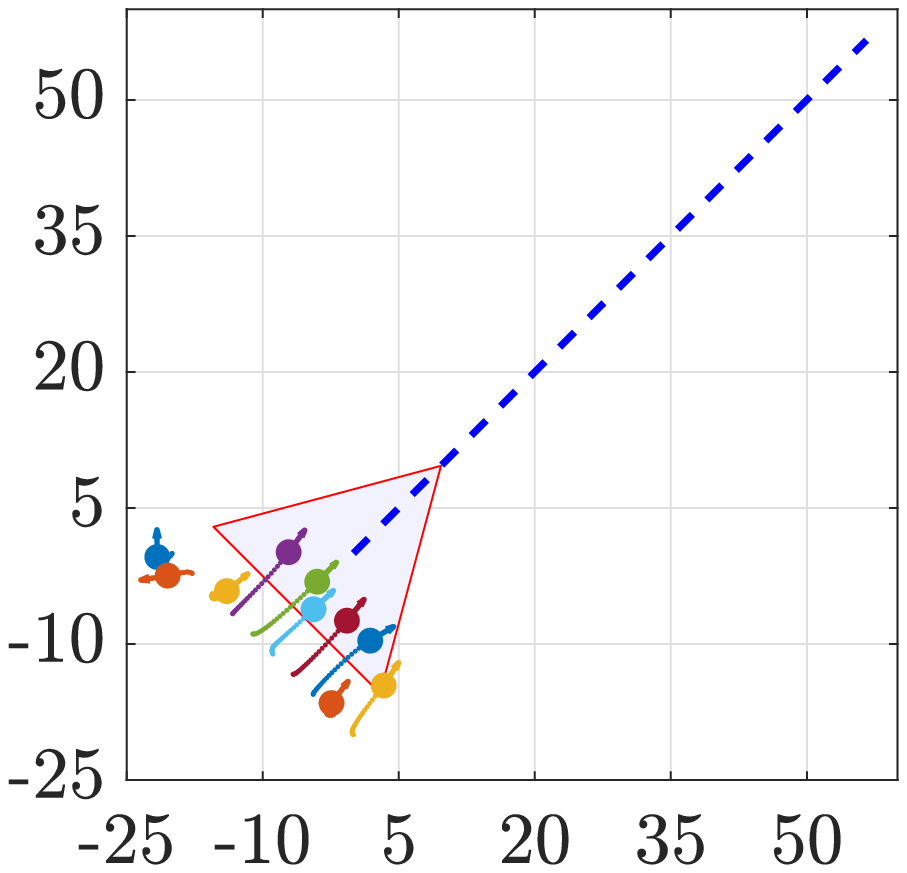}
\vspace{-2.5em}\caption{$t=2.4(s)$}
\label{fig:db_moving_triangle_2_4}
\end{subfigure}
\begin{subfigure}[b]{0.52\columnwidth}
\includegraphics[width=\columnwidth]{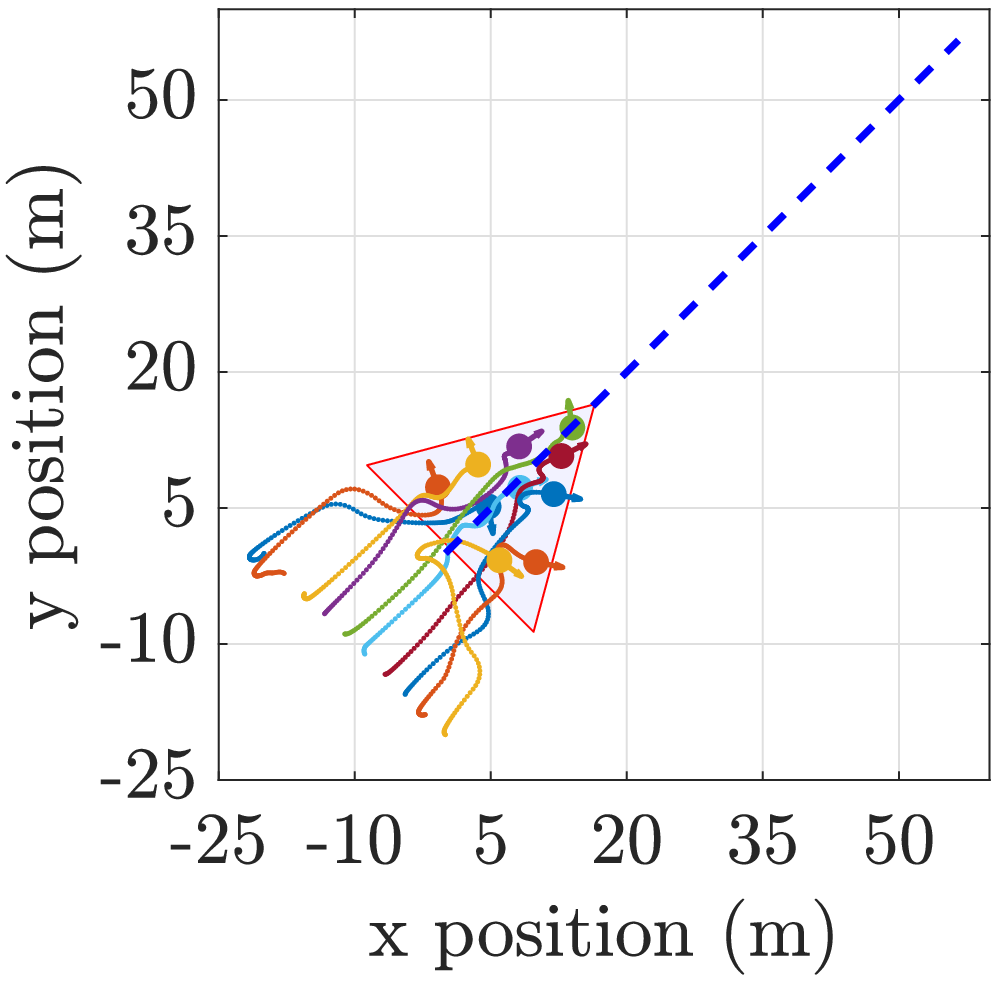}
\caption{$t=12(s)$}
\label{fig:db_moving_triangle_12}
\end{subfigure}
\hspace{-1.7em}
\begin{subfigure}[b]{0.52\columnwidth}
\includegraphics[width=\columnwidth]{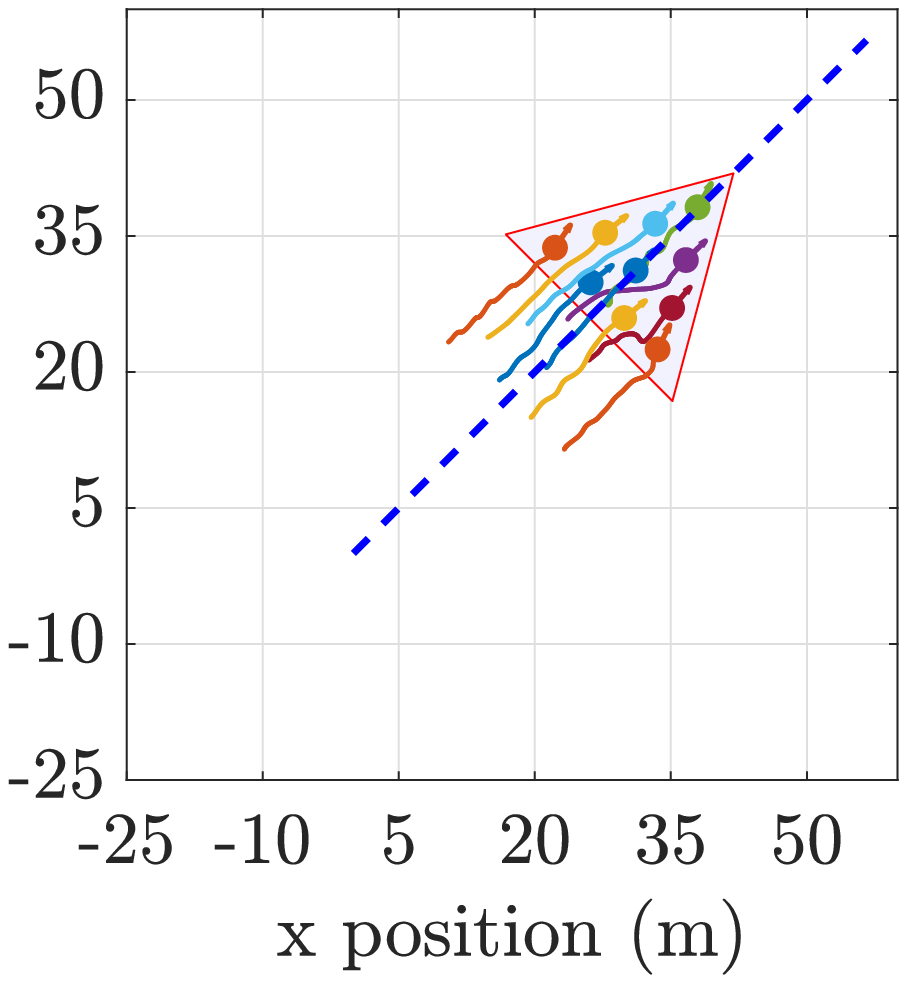}
\caption{$t=48(s)$}
\label{fig:db_moving_triangle_48}
\end{subfigure}
\caption{Vehicles with planar fixed-wing aircraft dynamics covering and following a moving equilateral triangular domain, when $N=10$, $c_{r}=2\,\text{(m)}$, $s_{max}=10\,\text{(m/s)}$, $s_{min}=0.5\,\text{(m/s)}$, $u_{\heading_{max}}=\pi/2\,\text{(rad/s)}$, $u_{\speed_{max}}=3 \,\text{(m/s$^{2})$}$, $\aI=1\,\text{(m/s$^{2})$}$, $\ah=2\,\text{(m/s$^{2})$}$, $\av=0.2\,\text{(m/s$^{2})$}$, $C_{al}=0.2\,\text{(m/s$^{2})$}$, $l_{al}=7.79\,\text{(m)}$, $v_{d}=\left(\frac{\sqrt{2}}{2},\frac{\sqrt{2}}{2}\right)\,\text{(m/s)}$, domain area $A=292.28\,\text{(m$^2$)}$ and $r_{d}=\sqrt{\frac{A}{N}}=5.4\,\text{(m)}$. Collision avoidance controller is included. The vehicles start in linear formation.}
\label{fig:db_moving_triangle}
\end{figure}

\medskip
\textbf{Domain moving in a circle.} The vehicles start in a line formation as shown in Figure \ref{fig:db_circle_0}. They reach the target domain and spread inside it (Figures \ref{fig:db_circle_12} and \ref{fig:db_circle_20}). Once they cover the domain, each of the fixed-wing agents follows a circular path with constant angular velocity $\omega=\frac{3\pi}{80}$ (Figure \ref{fig:db_circle_40}). Under this configuration the vehicles have reached their terminal speed and do not require extra acceleration, i.e. $u_{\speed}=0$, however they should maintain a turn rate of $u_{\heading}=\omega$.

The collision count goes from 7 to 1 by including collision avoidance. Similar to previous sections, we note that our approach based on pairwise collision avoidance does not guarantee safety when a vehicle has to avoid two or more vehicles at the same time.

\begin{figure}
\centering
\begin{subfigure}[b]{0.52\columnwidth}
\includegraphics[width=\columnwidth]{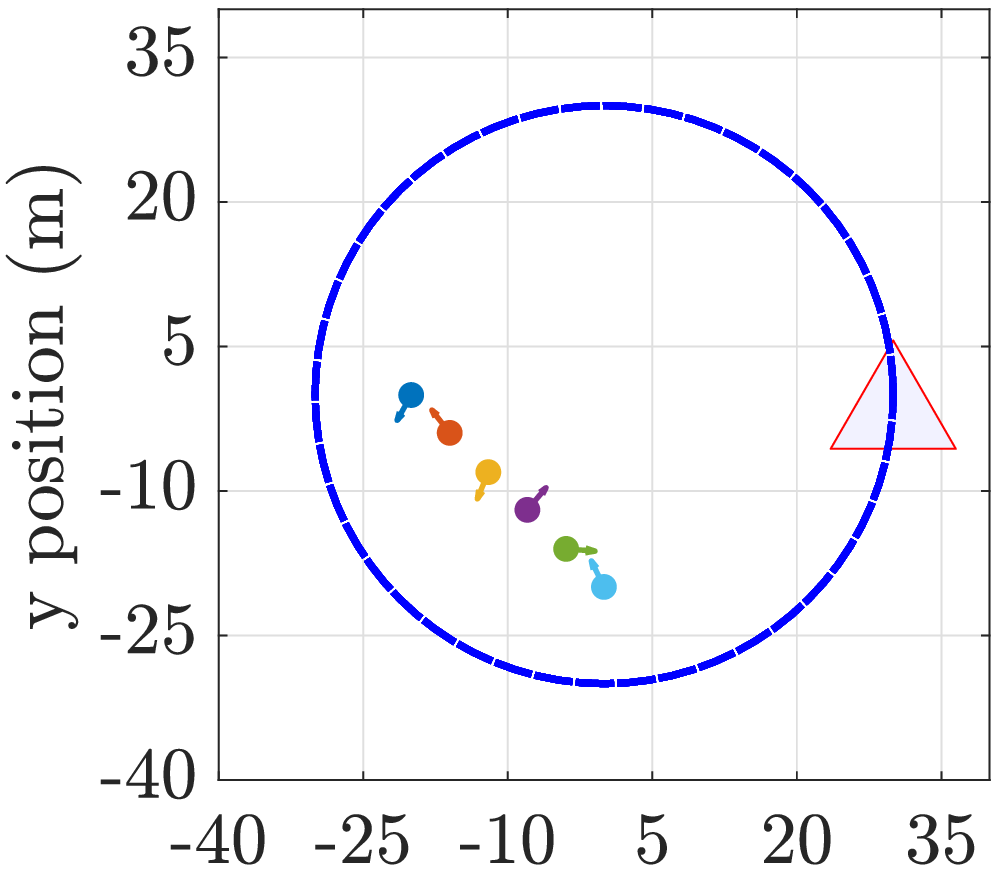}
\vspace{-2.5em}\caption{$t=0(s)$}
\label{fig:db_circle_0}
\end{subfigure}
\hspace{-1.7em}
\begin{subfigure}[b]{0.52\columnwidth}
\includegraphics[width=\columnwidth]{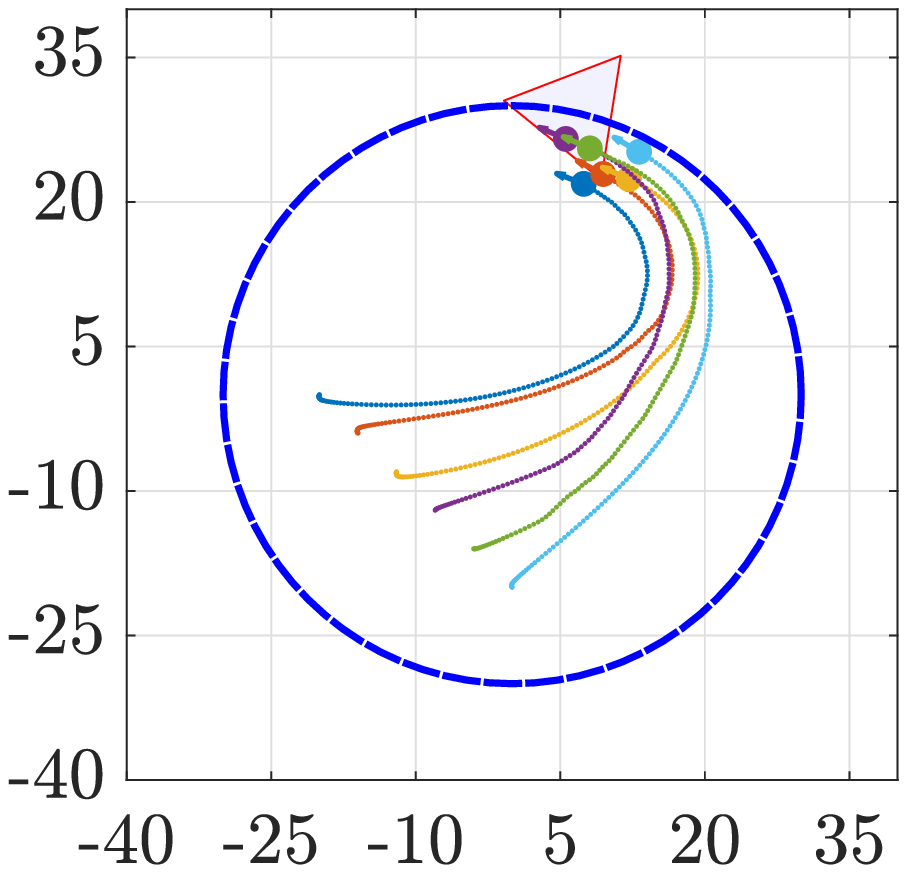}
\vspace{-2.5em}\caption{$t=12(s)$}
\label{fig:db_circle_12}
\end{subfigure}
\begin{subfigure}[b]{0.52\columnwidth}
\includegraphics[width=\columnwidth]{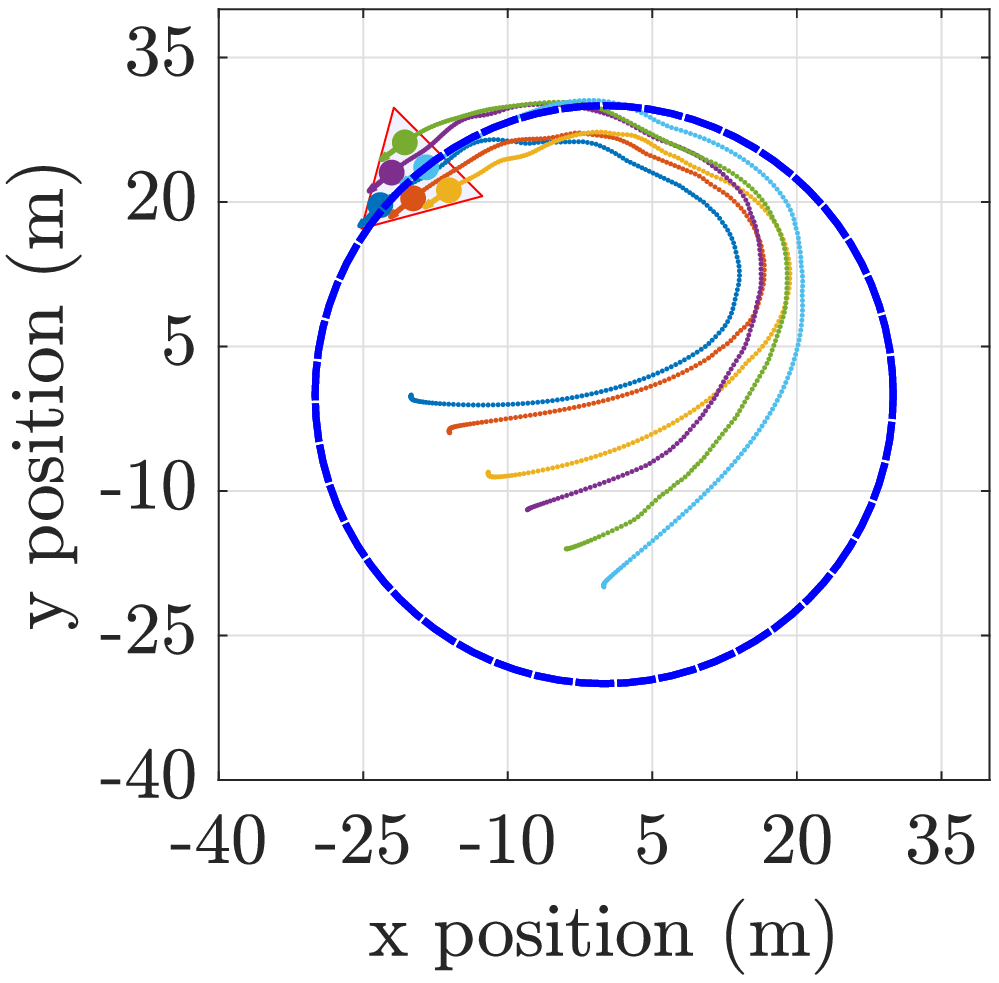}
\caption{$t=20(s)$}
\label{fig:db_circle_20}
\end{subfigure}
\hspace{-1.7em}
\begin{subfigure}[b]{0.52\columnwidth}
\includegraphics[width=\columnwidth]{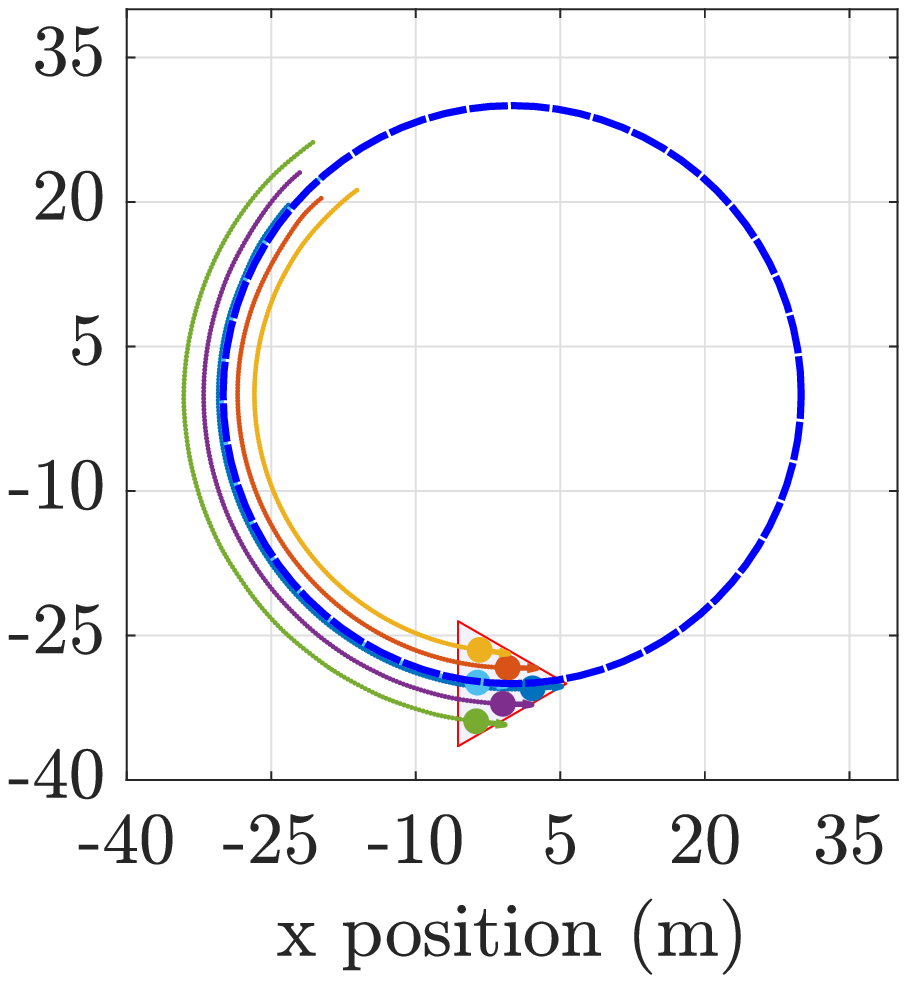}
\caption{$t=40(s)$}
\label{fig:db_circle_40}
\end{subfigure}

\caption{Vehicles with planar fixed-wing aircraft dynamics covering and following a non-zero acceleration triangular domain moving over the path $\left(30\cos\left(\frac{3\pi}{80} t\right),30\sin\left(\frac{3\pi}{80} t\right)\right)$. Here, $N=6$, $c_{r}=2\,\text{(m)}$, $\speed_{min}=0.5\,\text{(m/s)}$,  $\speed_{max}=5\,\text{(m/s)}$, $u_{\heading_ {max}}=\pi/2 \,\text{(rad/s)}$, $u_{\speed_{max}}=3 \,\text{(m/s$^2$)}$, $a_{v}=1.2$, $l_{al}=3.679$, $C_{al}=1.5$, $a_I=5$, $a_h=2.7$, domain area $A=73.07\,\text{(m$^2$)}$ and $r_{d}=\sqrt{\frac{A}{N}}=3.489\,\text{(m)}$. The vehicles start in a linear formation, approach and cover the domain, while following it. The collision avoidance controller is included.}
\label{fig:db_circle}
\end{figure}

%% file: conclusions.tex
\section{Conclusion}
\label{sect:conclusion}

\subsection{Summary of results}\label{subsec:summary_of_results}

Our proposed controller for multi-vehicle coordination  allows a swarm of vehicles to cover moving planar shapes.
Unlike previous coverage controllers that assumed first-order vehicle models, our coverage controllers use more realistic second-order models -- double integrator and fixed-wing aircraft.
We prove that our coverage controller achieves coverage and flocking with moving planar domains, and that the cover configurations of interest are locally asymptotically stable. Using HJ reachability analysis, we guarantee pairwise collision avoidance while accounting for bounded control inputs. In addition, we also derive the analytical solution to the associated HJ PDE for the double integrator model.

Our numerical simulations illustrate successful coverage of static and moving domains on four representative scenarios: static square, non-accelerated moving triangular and arrowhead (non-convex) domains, and a triangular domain following a circular path. 
While the first three scenarios are covered by our theoretical results, the last is not. Nevertheless, we find satisfactory numerical results in this case as well, suggesting some generality of the proposed technique.
For simulations involving the double integrator, we observe drastic reduction of collisions when using the HJ-based collision avoidance controller.

\subsection{Future work}\label{subsec:future_directions}

Immediate future work includes parameter tuning to reduce oscillations in the vehicles' movement, studying three-dimensional coverage, investigating geometrical properties of steady states, investigating scenarios involving partial information, and implementing our approach on robotic platforms.

%% file: biography.tex
\begin{IEEEbiography}[{\includegraphics[width=1in,height=1.25in,clip,keepaspectratio]{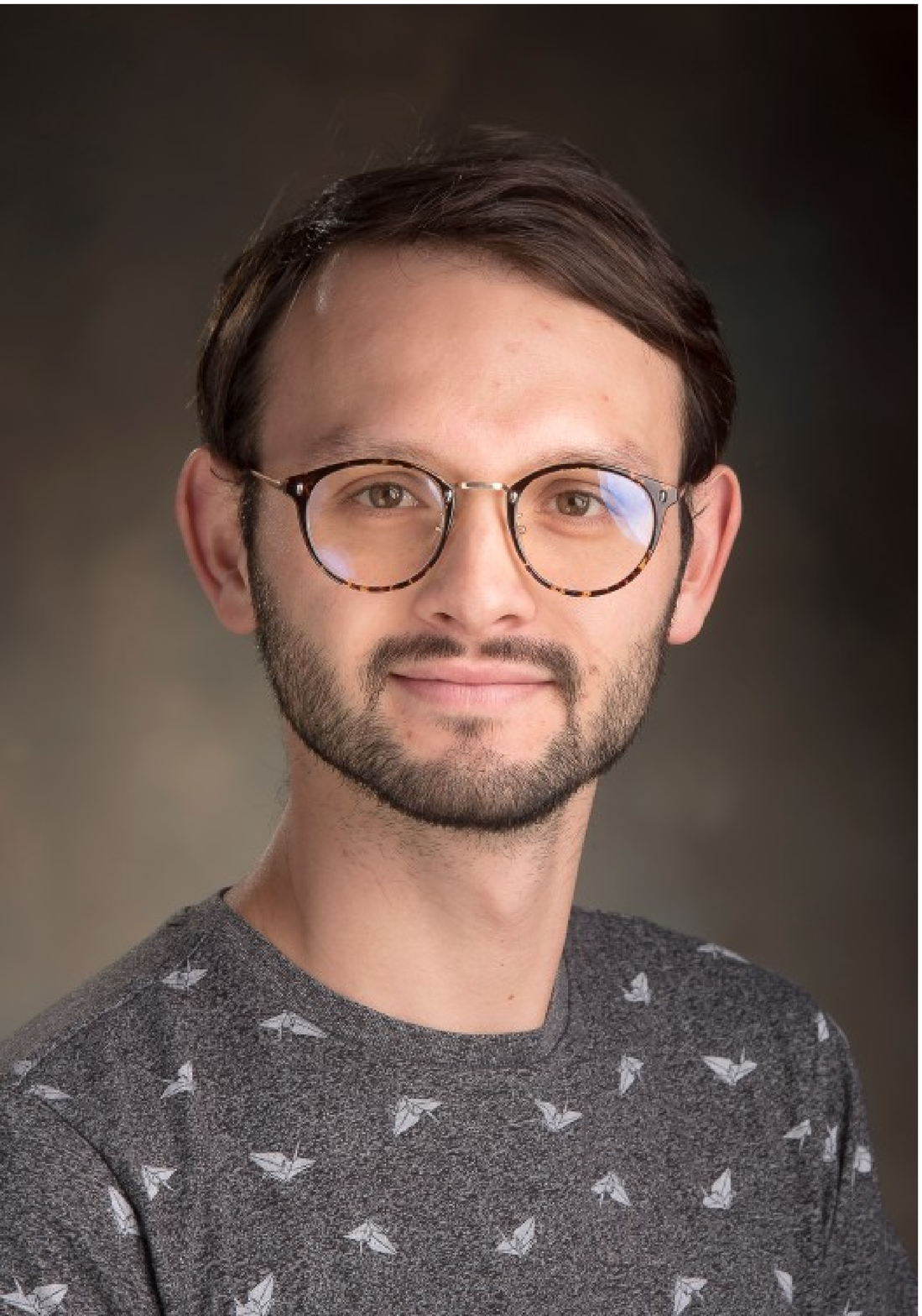}}]{Juan Chacon}
          is a recently graduated master's student from the department of mathematics at Simon Fraser University, Burnaby, BC, Canada, member of the Multi-Agent Robotic Systems Lab. He completed his MSc and BASc in the Mathematics Department at Universidad Nacional de Colombia, Bogota in 2013 and 2015 respectively. He received his BEng in Electronic Engineering from the Universidad Distrital Francisco Jose de Caldas in 2016. His research interests include multi-agent systems and aggregation models.
\end{IEEEbiography}
\begin{IEEEbiography}[{\includegraphics[width=1in,height=1.25in,clip,keepaspectratio]{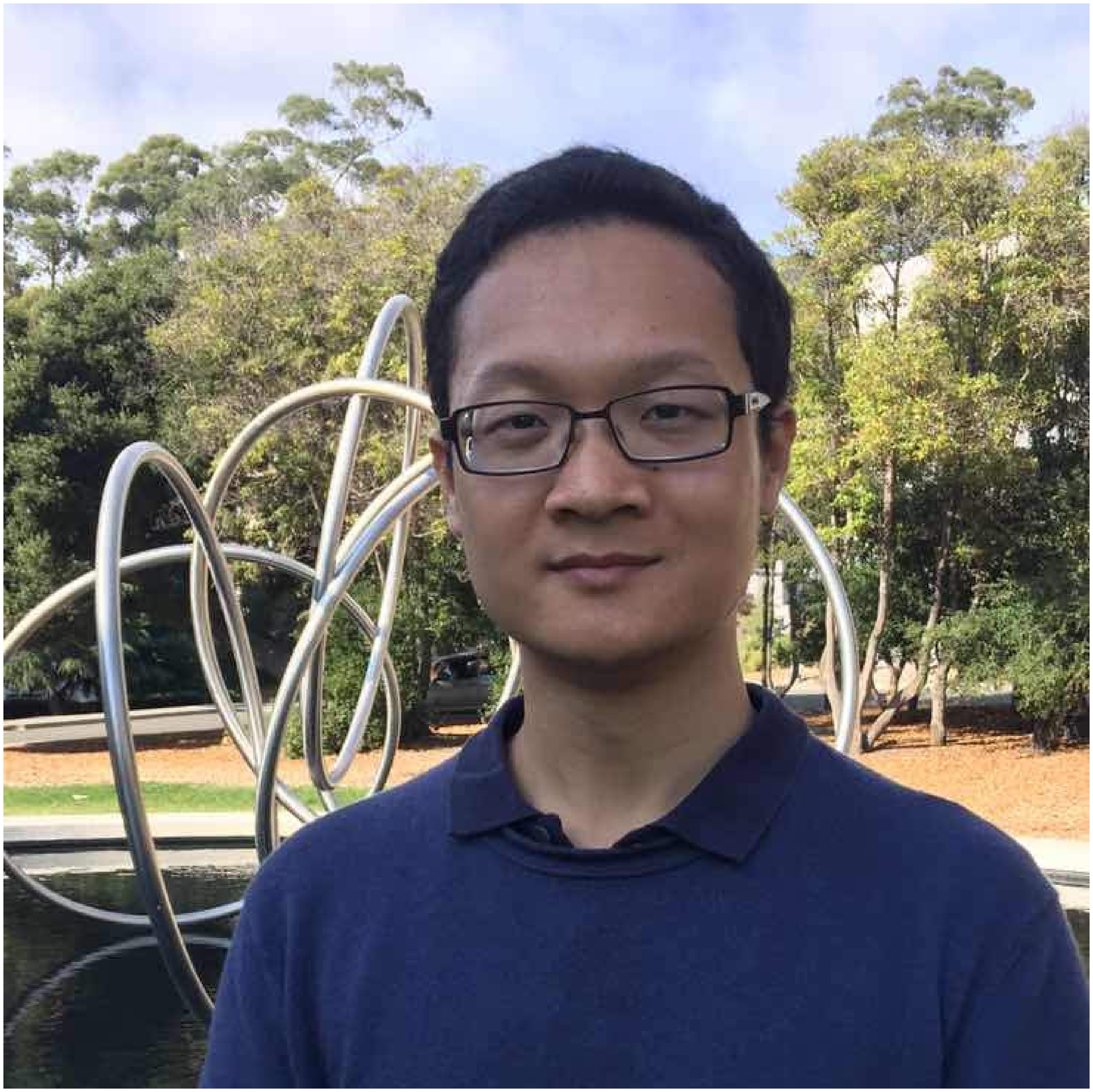}}]{Mo Chen}
          is an Assistant Professor in the School of Computing Science at Simon Fraser University, Burnaby, BC, Canada, where he directs the Multi-Agent Robotic Systems Lab. He completed his PhD in the Electrical Engineering and Computer Sciences Department at the University of California, Berkeley in 2017, and received his BASc in Engineering Physics from the University of British Columbia in 2011. From 2017 to 2018, Mo was a postdoctoral researcher in the Aeronautics and Astronautics Department in Stanford University. His research interests include multi-agent systems, safety-critical systems, and practical robotics.
\end{IEEEbiography}
\begin{IEEEbiography}[{\includegraphics[width=1in,height=1.25in,clip,keepaspectratio]{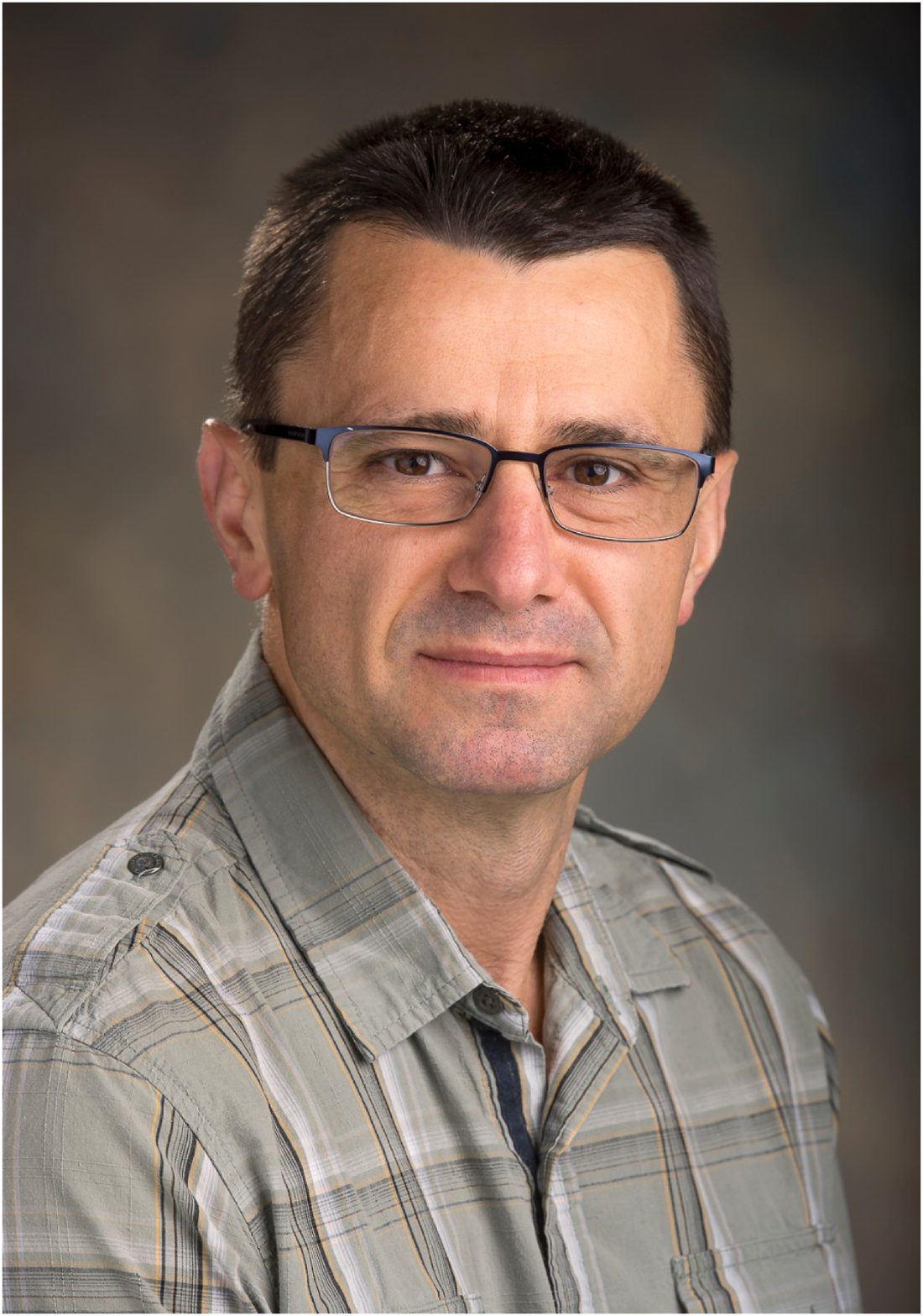}}]{Razvan Fetecau}
          is a Professor in the Department of Mathematics at Simon Fraser University, Burnaby, BC, Canada. He received a B.Sc. degree in Mathematics and Mechanics from A.I. Cuza University, Iasi, Romania in 1997, an M.Sc. degree in Mathematics from the University of Bucharest, Romania in 1998, and he completed his PhD in Applied and Computational Mathematics at California Institute of Technology, USA in 2003.  From 2003 to 2006, Razvan was a Szeg\"{o} Assistant Professor in the Department of Mathematics at Stanford University, USA. His research interests include dynamical systems and PDEs for self-collective and swarming behaviour, regularizations of fluid dynamics equations, and geometric mechanics.
\end{IEEEbiography}